\documentclass{article}

\usepackage[utf8]{inputenc}             
\usepackage[T1]{fontenc}                
\usepackage{hyperref}                   
\usepackage{url}                        
\usepackage{booktabs}                   
\usepackage{amsfonts}                   
\usepackage{nicefrac}                   
\usepackage{microtype}                  
\usepackage{xspace,color,colortbl}      
\usepackage{fullpage}

\usepackage{graphicx} 
\usepackage{caption}  
\usepackage[utf8]{inputenc}
\usepackage{pgfplots}
    \pgfplotsset{compat=1.18}
\usepackage{tikz}
    \usetikzlibrary{shapes,decorations}
    \usetikzlibrary{positioning}
    \usetikzlibrary{external} 
\usepackage{amsmath, amsfonts, amsthm, mathrsfs, mathtools}
\usepackage{nicefrac}
\usepackage{MnSymbol}
\usepackage[ruled, noline,noend]{algorithm2e} 
\usepackage{enumitem}
\usepackage{hyperref}
\usepackage[numbers,sort]{natbib}
\usepackage[capitalize,noabbrev]{cleveref} 

\newcommand{\cD}{\mathcal{D}}
\newcommand{\cE}{\mathcal{E}}
\newcommand{\cF}{\mathcal{F}}

\newcommand{\cH}{\mathcal{H}}

\newcommand{\E}{\mathbb{E}}

\newcommand{\I}{\mathbb{I}}

\newcommand{\N}{\mathbb{N}}

\newcommand{\Pb}{\mathbb{P}}
\newcommand{\bbR}{\mathbb{R}}

 \newcommand{\e}{\varepsilon}

\newcommand{\lrb}[1]{\left(#1\right)}
\newcommand{\brb}[1]{\bigl(#1\bigr)}

\newcommand{\lsb}[1]{\left[#1\right]}
\newcommand{\bsb}[1]{\bigl[#1\bigr]}
\newcommand{\Bsb}[1]{\Bigl[#1\Bigr]}
\newcommand{\bbsb}[1]{\biggl[#1\biggr]}
\newcommand{\lcb}[1]{\left\{#1\right\}}
\newcommand{\bcb}[1]{\bigl\{#1\bigr\}}

\newcommand{\labs}[1]{\left\lvert#1\right\rvert}

\DeclareMathOperator*{\argmax}{argmax}
\newcommand{\dif}{\,\mathrm{d}}

\newcommand{\fracc}[2]{#1/#2}

\newcommand{\mypapertitle}{Trading Volume Maximization with Online Learning}

\DeclareSymbolFont{extraup}{U}{zavm}{m}{n}
\DeclareMathSymbol{\clubsuit}{\mathalpha}{extraup}{84}
\DeclareMathSymbol{\spadesuit}{\mathalpha}{extraup}{81}
\DeclareMathSymbol{\varheartsuit}{\mathalpha}{extraup}{86}
\DeclareMathSymbol{\vardiamondsuit}{\mathalpha}{extraup}{87}

\newcommand{\gft}{\mathrm{g}}
\newcommand{\GFT}{\mathrm{G}}

\newcommand{\ftpsi}{FE$\Psi$}
\newcommand{\fem}{FEM}

\newtheorem{theorem}{Theorem}
\newtheorem{lemma}{Lemma}

\title{\mypapertitle}

\author{%
    \textbf{Tommaso Cesari} \\
    University of Ottawa, Ottawa, Canada
    \\
    \textbf{Roberto Colomboni} \\
    Politecnico di Milano \& Universit\`a degli Studi di Milano, Milano, Italy
    \\
}

\begin{document}

\maketitle

\begin{abstract}
We explore brokerage between traders in an online learning framework.
At any round $t$, two traders meet to exchange an asset, provided the exchange is mutually beneficial.
The broker proposes a trading price, and each trader tries to sell their asset or buy the asset from the other party, depending on whether the price is higher or lower than their private valuations.
A trade happens if one trader is willing to sell and the other is willing to buy at the proposed price.

Previous work provided guidance to a broker aiming at enhancing traders' total earnings by maximizing the \emph{gain from trade}, defined as the sum of the traders' net utilities after each interaction.
In contrast, we investigate how the broker should behave to maximize the trading volume, i.e., the \emph{total number of trades}.

We model the traders' valuations as an i.i.d.\ process with an unknown distribution.

If the traders' valuations are revealed after each interaction (full-feedback), and the traders' valuations cumulative distribution function (cdf) is continuous, we provide an algorithm achieving logarithmic regret and show its optimality up to constant factors.

If only their willingness to sell or buy at the proposed price is revealed after each interaction ($2$-bit feedback), we provide an algorithm achieving poly-logarithmic regret when the traders' valuations cdf is Lipschitz and show that this rate is near-optimal.

We complement our results by analyzing the implications of dropping the regularity assumptions on the unknown traders' valuations cdf. 
If we drop the continuous cdf assumption, the regret rate degrades to $\Theta(\sqrt{T})$ in the full-feedback case, where $T$ is the time horizon. 
If we drop the Lipschitz cdf assumption, learning becomes impossible in the $2$-bit feedback case. 
\end{abstract}

\textbf{Keywords:} Regret minimization, Online learning, Two-sided markets

\section{Introduction}

In modern financial markets, Over-the-Counter (OTC) trading platforms have emerged as dynamic and decentralized hubs, offering diverse alternatives to traditional exchanges.
These markets have experienced remarkable growth, solidifying their central role in the global financial ecosystem: OTC asset trading in the US surpassed 50 trillion USD in value in 2020 \cite{weill2020search}, with an upward trend documented since 2016 \cite{bis2023}.

Brokers play a crucial role in OTC markets.
Beyond acting as intermediaries between traders, they utilize their understanding of the market to identify the optimal prices for assets.
Additionally, traders in these markets often respond to price changes: higher prices usually lead to selling, while lower prices typically result in buying \cite{sherstyuk2020randomized}.
This adaptability appears across various asset classes, including stocks, derivatives, art, collectibles, precious metals and minerals, energy commodities (like gas and oil), and digital currencies (cryptocurrencies) \cite{bolic2023online}.

Our study draws inspiration from recent research analyzing the bilateral trade problem from an online learning perspective \cite{cesa2021regret, azar2022alpha, cesa2023bilateral, cesa2023repeated, bolic2023online, bernasconi2023no}. In particular, we build on insights from \cite{bolic2023online}, which addresses the brokerage problem in OTC markets where traders may decide to buy or sell their assets depending on prevailing market conditions.

Previous work has entirely focused on scenarios where brokers aim at maximizing traders' earnings, measured by cumulative \emph{gain from trade}—--the total net utilities of traders over the entire sequence of interactions with the broker.
This could come at the expense of sacrificing potential trades where the net gain of the traders is small, which could be critical for traders making a living off of small margins.
From the broker's perspective, too, it might not be as beneficial to potentially miss out on traders' exchanges, given that, typically, brokers only earn when trades occur.
For this reason, our approach adopts a different perspective, aiming at providing strategies that boost the trading volume by maximizing the \emph{number of trades} within the entire broker-traders interaction sequence.

\subsection{Setting}
In what follows, for any two real numbers $a,b$, we denote their minimum by $a\wedge b$ and their maximum by $a \vee b$.
We now describe the brokerage online learning protocol.

For any time $t = 1,2,\dots$
\begin{itemize}
    \item[$\bullet$] Two traders arrive with their private valuations $V_{2t-1}$ and $V_{2t}$
    \item[$\bullet$] The broker proposes a trading price $P_t$
    \item[$\bullet$] If the price \(P_t\) is between the lowest valuation \(V_{2t-1} \wedge V_{2t}\) and the highest valuation \(V_{2t-1} \vee V_{2t}\)---meaning the trader with the lower valuation is willing to sell at \(P_t\) and the trader with the higher valuation is willing to buy at \(P_t\)---the transaction occurs with the higher-valuation trader purchasing the asset from the lower-valuation trader at the price \(P_t\)
    \item[$\bullet$] The broker receives some feedback
\end{itemize}

As commonly assumed in the existing bilateral trade literature, we assume valuations and prices belong to $[0,1]$.
While previous literature aims at maximizing the cumulative \emph{gain from trade}---defined as the sum of traders' net utilities in the whole interaction sequence---our objective is to maximize the \emph{number of trades}. 
Formally, for any $p,v_1,v_2 \in [0,1]$, our utility posting a price $p$ when the valuations of the traders are $v_1$ and $v_2$ is
\[
	\gft(p,v_1,v_2) \coloneqq \I \lcb{ v_1 \wedge v_2 \le p \le v_1 \vee v_2 } \;.
\]
The goal of the broker is to minimize the \emph{regret}, defined, for any time horizon $T\in\N$, as
\[
	R_T \coloneqq \sup_{p\in[0,1]} \E \lsb{ \sum_{t=1}^T \brb{\GFT_t(p) - \GFT_t(P_t)} }\;,
\]
where $\GFT_t(q) \coloneqq \gft(q, V_{2t-1}, V_t)$ for all $q\in[0,1]$ and $t \in \N$, and the expectation is taken over the randomness present in $(V_t)_{t\in \N}$ and the (possible) randomness used by the broker's algorithm to generate the prices $(P_t)_{t\in \N}$.

As in \cite{bolic2023online}, we assume that traders' valuations $V, V_1, V_2,\dots$ are generated i.i.d.\ from an \emph{unknown} distribution $\nu$---a practical assumption for large and stable markets.

Finally, we consider the following two different types of feedback commonly studied in the online learning bilateral trade literature:

\begin{itemize}
    \item \emph{Full-feedback.} At each round $t$, after having posted the price $P_t$, the broker has access to the traders' valuations $V_{2t-1}$ and $V_{2t}$.
    \item \emph{$2$-bit feedback.} At each round $t$, after having posted the price $P_t$, the broker has access to the indicator functions $\I\{ V_{2t-1} \le P_t\}$ and  $\I\{ V_{2t} \le P_t\}$.
\end{itemize}

The full-feedback model draws its motivation from \emph{direct revelation mechanisms}, where the traders disclose their valuations $V_{2t-1}$ and $V_{2t}$ before each round, but the mechanism has access to this information only after having posted the current bid $P_t$ \cite{cesa2021regret,cesa2023bilateral}.

The $2$-bit feedback model corresponds to \emph{posted price} mechanisms, where the broker has access only to the traders' willingness to buy or sell at the proposed posted price, and the valuations $V_{2t-1}$ and $V_{2t}$ are \emph{never} revealed.

\subsection{Overview of Our Contributions} 
In the full-feedback case, if the distribution $\nu$ of the traders' valuations has a \emph{continuous} cdf, we design an algorithm (\Cref{a:ftm}) suffering $O ( \ln T )$ regret in the time horizon $T$ (\Cref{t:ftm}), and we provide a matching lower bound (\Cref{t:lower-bound-full-M}).
We complement these results by showing that dropping the continuous cdf assumption leads to a worse regret rate of $\Omega(\sqrt{T})$ (\Cref{t:lower-full-info-sqrtT}), and we design an algorithm (\Cref{a:ftpsi}) achieving $O(\sqrt{T})$ regret (\Cref{t:ftpsi}).

In the $2$-bit feedback case, if the cdf of the traders' valuations is $M$-Lipschitz, we design an algorithm (\Cref{a:MBS}) achieving regret $ O \brb{ \ln(M T)\ln T }$ (\Cref{t:mbs}) where $T$ is the time horizon, and provide a near-matching lower bound $\Omega \brb{ \ln{(M T)} }$ (\Cref{t:lower-bound-2-bit-M}). 
We complement these results by showing that the problem becomes unlearnable if we drop the Lipschitzness assumption (\Cref{t:lower-bound-2-bit-linear}).

For a full summary of our results, see \Cref{table:results}.
{
\renewcommand{\arraystretch}{1.4}
\begin{table}[tbp]
\centering
\begin{tabular}{c|c|c|c|}
    \cline{2-4}
    & $M$-Lipschitz & Continuous & General                        \\ \hline
    \multicolumn{1}{|l|}{Full} &  \cellcolor[HTML]{00dfff} $\Omega(\ln T)$ \  Thm \ref{t:lower-bound-full-M} & \cellcolor[HTML]{00dfff} $O(\ln T)$ \  Thm \ref{t:ftm} & \cellcolor[HTML]{eeff00} $\Theta\brb{ \sqrt{T} }$ \  Thm \ref{t:lower-full-info-sqrtT}+\ref{t:ftpsi} 
    \\ \hline
    \multicolumn{1}{|l|}{$2$-Bit}   & \cellcolor[HTML]{35ff41} $O \brb{ \ln(M T)\ln T }$, $\Omega \brb{ \ln(M T) }$ \ Thm \ref{t:mbs}+\ref{t:lower-bound-2-bit-M} & \cellcolor[HTML]{FD6864} $\Omega(T)$ \ Thm \ref{t:lower-bound-2-bit-linear} & \cellcolor[HTML]{FD6864} $\Omega(T)$  
    \\ \hline
\end{tabular}
\caption{\footnotesize{Overview of all the regret regimes: $\ln T$ (cyan), $\ln(MT)$ (green), $\sqrt{T}$ (yellow), and $T$ (red), depending on the feedback (full or $2$-bit) and the assumption on the cdf ($M$-Lipschitz, continuous, or no assumptions).
}}
\label{table:results}
\end{table}
}

\subsection{Techniques and Challenges}
\label{s:tech-and-challenges}
Online learning with a continuous action domain and full-feedback is usually tackled by discretizing the action domain and then playing an optimal expert algorithm on the discretization, or by directly running exponential weights algorithms in the continuum \cite{MaillardM10, krichene2015hedge,cesa2024regret}.
These approaches require that the (expected) reward function is Lipschitz and lead to a regret rate of order $\tilde{O}(\sqrt{T})$.
In contrast, our expected reward function is \emph{not} Lipschitz in general.
To overcome this challenge, we leverage the specific structure of the problem by proving \Cref{l:mysticus}, which enables us to design an algorithm that achieves an exponentially better regret rate of $O(\ln T)$ even when the underlying cdf---and hence the associated reward function---is only continuous.
Moreover, we establish a matching $\Omega(\ln T)$ lower bound that, surprisingly, applies even when the reward function is Lipschitz, demonstrating that additional Lipschitz regularity beyond continuity does not contribute to faster rates in this setting.
This lower bound construction is particularly challenging because the shape of the function $p \mapsto \E\bsb{G_t(p)}$ can only be controlled indirectly through the traders' valuation distribution: to avoid exceedingly complex calculations, extra care is required in selecting appropriate instances. Even then, we needed a subtle and somewhat intricate Bayesian argument to obtain the lower bound.

In the $2$-bit feedback model, we remark that the available feedback is enough to reconstruct \emph{bandit} feedback.
Consequently, when the underlying cdf---and hence the expected reward function---is $M$-Lipschitz, a viable approach is to discretize the action space $[0,1]$ with $K$ uniformly spaced points and run an optimal bandit algorithm on the discretization.
This approach immediately yields a regret rate of order $O(\fracc{MT}{K} + \sqrt{KT})$.
This bound leads to a regret of order $O(M^{1/3} T^{2/3})$ by tuning $K \coloneqq \Theta(M^{2/3}T^{1/3})$ when $M$ is known to the learner, or of order $O(M T^{2/3})$ by tuning $K \coloneqq \Theta(T^{2/3})$ when the learner does not possess this knowledge. 
In contrast, we exploit the extra information provided by the $2$-bit feedback and the intuition provided by \Cref{l:mysticus} to devise a binary search algorithm achieving the exponentially better rate of $O\brb{\ln(MT)\ln T}$, with the additional feature of being oblivious to $M$.
Our corresponding lower bound shows that this rate is optimal (up to a $\ln T$ factor), demonstrating through an information-theoretic argument that some sort of binary search is essentially a necessary step for optimal learning.

\subsection{Related Work}
\label{s:related}

Since the pioneering work of Myerson and Satterthwaite \cite{myerson1983efficient}, the study of bilateral trade has grown significantly, particularly from a game-theoretic and approximation perspective \cite{Colini-Baldeschi16,Colini-Baldeschi17,BlumrosenM16,brustle2017approximating,colini2020approximately,babaioff2020bulow,dutting2021efficient,DengMSW21,kang2022fixed,archbold2023non}. For a comprehensive overview, refer to \cite{cesa2023bilateral}.

In recent years, the focus has expanded to include online learning settings for bilateral trade. Given their close relevance to our paper, we concentrate our discussion on these works.

In \cite{cesa2021regret,azar2022alpha,cesa2023bilateral,cesa2023repeated,bernasconi2023no,cesa2024regret}, the authors examined bilateral trade problems where the reward function is the \emph{gain from trade} and each trader has a fixed role as either a seller or a buyer.

In \cite{cesa2021regret}, the authors investigated a scenario where seller and buyer valuations form two distinct i.i.d.\ sequences.
In the full-feedback case, they achieved a regret bound of $\widetilde{O}(\sqrt{T})$, which was later refined to $O(\sqrt{T})$ in \cite{cesa2023bilateral}.
They also demonstrated a worst-case regret of $\Omega(\sqrt{T})$ even when sellers' and buyers' valuations are independent of each other and their cdfs are Lipschitz.
For the $2$-bit feedback scenario under i.i.d.\ valuations, \cite{cesa2021regret} proved that any algorithm must suffer linear regret, even under either the $M$-Lipschitz joint cdf assumption or the traders' valuation independence assumption.
However, when both conditions are simultaneously satisfied, they proposed an algorithm achieving a regret rate of $\widetilde{O}(M^{1/3} T^{2/3})$, later refined to $O(M^{1/3} T^{2/3})$ in \cite{cesa2023bilateral}.
\cite{cesa2021regret} also established a worst-case regret lower bound of $\Omega(T^{2/3})$ in this case, which, however, does not display any dependence on $M$. 

\cite{cesa2021regret,cesa2023bilateral} also showed that the adversarial bilateral trade problem is unlearnable even with full-feedback.
To achieve learnability beyond the i.i.d.\ case, \cite{cesa2023repeated,cesa2024regret} explored weakly budget-balanced mechanisms, allowing the broker to post different selling and buying prices as long as the buyer pays more than what the seller receives.
They demonstrated that learning can be achieved using weakly budget-balanced mechanisms in the $2$-bit feedback setting at a regret rate of $\widetilde{O}(M T^{3/4})$ when the joint seller/buyer cdf may vary over time but is $M$-Lipschitz. 
Furthermore, for the same setting, they provided a $\Omega(T^{3/4})$ matching lower bound in the time horizon, even when the process is required to be i.i.d., 
but their lower bound does not feature any dependence on $M$.
\cite{azar2022alpha} investigated whether learning is possible in the adversarial case by considering $\alpha$-regret, achieving $\widetilde{\Theta}(\sqrt{T})$ bounds for $2$-regret in full-feedback, and a $\widetilde{O}(T^{3/4})$ upper bound in $2$-bit feedback.
Following another direction, \cite{bernasconi2023no} explored globally budget-balanced mechanisms in the adversarial case, showing a $\Theta(\sqrt{T})$ regret rate in full-feedback and a $\widetilde{O}(T^{3/4})$ rate in the $2$-bit feedback case.

The closest to our work is \cite{bolic2023online}, where the authors studied the same i.i.d.\ version of our trading problem with flexible seller and buyer roles, but with the target reward function being the \emph{gain from trade}.
Under the $M$-Lipschitz cdf assumption, they obtained tight $\Theta(M \ln T)$ regret in the full-feedback case.
Surprisingly, in the same full-feedback case, but using our different reward function, we achieve a $\Theta(\ln T)$ regret rate even when the cdf is only continuous: in our case, the additional Lipschitz regularity does not offer any speedup once the continuity assumption is fulfilled.
Furthermore, under the $M$-Lipschitz cdf assumption, \cite{bolic2023online} proved a sharp rate of $\Theta(\sqrt{MT})$ in the $2$-bit feedback case.
Interestingly, using our different reward function, we achieve an exponentially faster upper bound of $O\brb{ \ln(MT)\ln T }$, which is tight up to a $\ln T$ factor.
If the Lipschitz cdf assumption is removed, the learning rate for both our problem and the one in \cite{bolic2023online} degrades to $\Theta(\sqrt{T})$ in the full-feedback case, and the problem becomes unlearnable in the $2$-bit feedback case.

\subsection{Limitations}
\label{s:limitations}

This paper investigates a setting where traders lack rigid buyer/seller roles. 
Our results do not apply to the case where there are two separate groups of traders, buyers and sellers, whose valuations are drawn from two distinct distributions. 
Technically speaking, the reason why our results do not apply to this different setting is that the key \Cref{l:mysticus} does not hold anymore.

Another conceivable limitation of this work is our focus on an i.i.d.\ setting.
However, in addition to i.i.d.\ settings being wildly popular for their ability to lead to solutions of great practical importance, this limitation is further mitigated by our impossibility result (\Cref{t:lower-bound-2-bit-linear}), which proves in particular that learning in an adversarial environment is impossible in $2$-bit feedback settings.

Lastly, we obtain fast regret rates only under regularity assumptions on the cdf of the traders' valuations. 
However, as we show in \Cref{s:beyond-regularity}, this issue is unavoidable: in the full-feedback setting, dropping the continuity assumption leads to an immediate exponential worsening of the minimax regret (from $\ln T$ to $\sqrt{T}$), while, in the $2$-bit feedback setting, dropping the Lipschitzness assumption makes learning entirely impossible, even if the cdf remains continuous.

\section{The Median Lemma}
In this section, we present the Median Lemma (\Cref{l:mysticus}), a simple but crucial result for what follows, and the key upon which our main algorithms are based.
At a high level, \Cref{l:mysticus} states that a broker who aims at maximizing the number of trades should post prices that are as close as possible to the median of the (unknown) traders' valuation distribution $\nu$, and the instantaneous regret which the broker incurs by playing any price is (proportional to) the square of the distance between the median and the price, if distances are measured with respect to the pseudo-metric induced by the traders' valuation cdf.
\begin{lemma}[The median lemma]
    \label{l:mysticus}
    If the cdf $F$ of $\nu$ is continuous, then, for any $t \in \N$ and any $p \in [0,1]$,
    \[
        \E\bsb{ \GFT_t(p) } = 2F(p)\brb{1-F(p)}
        \qquad\text{and}\qquad
        \frac{1}{2}-\E\bsb{ \GFT_t(p) } = 2\lrb{\frac{1}{2}-F(p)}^2\;.
    \]
    In particular, the function $p \mapsto \E\bsb{ \GFT_t(p) }$ is maximized at any point $m \in [0,1]$ such that $F(m) = \frac{1}{2}$.
\end{lemma}
Before presenting the proof of \Cref{l:mysticus}, we just remark that points $m \in [0,1]$ satisfying $F(m) = 1/2$ do exist by the intermediate value theorem, because $F(0)=0$, $F(1)=1$, and $F$ is continuous.

\begin{proof}
For each $t \in \N$ and each $p \in [0,1]$, we have that
\begin{align*}
    \E\bsb{\GFT_t(p)}
&=
    \Pb\Bsb{\bcb{V_{2t-1} \le p < V_{2t}} \cup \bcb{V_{2t} \le p \le V_{2t-1}}}
\\
&=
    \Pb\bsb{ V_{2t-1} \le p} \Pb\bsb{ p < V_{2t}}
    +
    \Pb\bsb{ V_{2t} \le p} \Pb\bsb{ p \le V_{2t-1}}
=
    2F(p)\brb{1-F(p)}\;,
\end{align*}
where the second equality follows from additivity and independence, while in the last equality we leveraged the continuity of $F$ to obtain $\Pb[p \le V_{2t-1}] = \Pb[p < V_{2t-1}] = 1 - F(p)$.
To conclude, it is enough to note that, for each $p \in [0,1]$ it holds that
$
    \nicefrac{1}{4} - F(p)\brb{1-F(p)}
=
    \brb{\nicefrac{1}{2} - F(p)}^2
$.
\end{proof}

Before concluding this section, we remark that \Cref{l:mysticus} plays an analogous role to the one played by the Approximation and Representation Lemmas in \cite{bolic2023online}, which together implied that a broker aiming at maximizing the \emph{gain from trade} should post prices that are close as possible to the \emph{mean} of the traders' valuations.
It is also interesting to note that in order to control the instantaneous regret by a squared distance, \cite{bolic2023online} needed to impose the Lipschitzness of the traders' valuation cdf, which is an unnecessary condition to obtain a corresponding result in our problem.

\section{Full-Feedback}

We now investigate how the broker should behave to maximize the number of trades in the full-feedback case where after each interaction the traders' valuations are disclosed.
We begin by studying the full-feedback case under the continuous cdf assumption.
In this case, taking inspiration from \Cref{l:mysticus}, a natural strategy is to play the \emph{empirical median}, which leads to \cref{a:ftm}.
\begin{algorithm}
Post $P_1 \coloneqq 1/2$ and receive feedback $V_{1}$, $V_{2}$\;
\For
{%
    time $t=2,3,\ldots$
}
{
    Post $P_t \coloneqq \frac{1}{2} \lrb{V^{(t-1)}_{2(t-1)}+V^{(t)}_{2(t-1)}}$, where $V^{(1)}_{2(t-1)},\dots,V^{(2(t-1))}_{2(t-1)}$ are the order statistics of the observed sample $V_1,\dots,V_{2(t-1)}$, and receive feedback $V_{2t-1}$, $V_{2t}$\;
}
\caption{Follow the Empirical Median (FEM)}
\label{a:ftm}
\end{algorithm}

The next theorem leverages \Cref{l:mysticus} to show that \Cref{a:ftm} suffer regret $O(\ln T)$ when the traders' valuation cdf is continuous.
\begin{theorem}
\label{t:ftm}
If $\nu$ has a continuous cdf $F$, the regret of \fem{} satisfies, for all time horizons $T\in \N$,
\[
    R_T
\le
    \frac{1}{2} + \frac{\pi}{2} \brb{ 1 + \ln (T-1) } \;.
\]
\end{theorem}

\begin{proof}
Without loss of generality, we can (and do!) assume that $T\ge 2$.
Then, we have
\begin{align*}
    R_{T}
\le
    \frac12+\max_{p \in [0,1]}
    \E\lsb{\sum_{t=2}^{T}\GFT_t(p)} -\E\lsb{\sum_{t=2}^{T}\GFT_t(P_t)}   
=
    \frac12+2\cdot\sum_{t=2}^{T}\E\lsb{\lrb{\frac{1}{2}-F(P_t)}^2}   
\end{align*}
Now, let $m \in [0,1]$ be such that $F(m) = 1/2$, and let $V$ be a random variable whose distribution is $\nu$, independent of $V_1,V_2,\dots$.
Then, for any $t \in \N$ such that $t \ge 2$ we have
\begin{align*}
    \E\lsb{\lrb{\frac{1}{2}-F(P_t)}^2}
&
=
    \E\lsb{\lrb{\Pb[m \le V \le P_t \mid P_t]}^2}
    +
    \E\lsb{\lrb{\Pb[ P_t \le V \le m\mid P_t]}^2}
\eqqcolon
    (I) + (II)\;.
\end{align*}
Now, for the term $(I)$, leveraging the fact that $V$ and $P_t$ are independent of each other, together with the Minkowski's integral inequality (see, e.g., \cite[Appendix A.1]{stein1970singular}), we have:
\begin{align*}
    \sqrt{(I)}
&=
    \sqrt{\E\lsb{\lrb{\E\bsb{\I\{m \le V \le P_t \}\mid P_t}}^2}}
\le
    \E \lsb{ \sqrt{ \E \lsb{ \brb{\I\{m \le V \le P_t\}}^2 \mid V }}}
\\
&
=
    \E \lsb{ \sqrt{ \Pb[m \le V \le P_t \mid V] }}
=
    \int_{[m,1]} \sqrt{\Pb[x \le P_t]} \dif \Pb_{V}(x)
=
    \int_{[m,1]} \sqrt{\Pb[x \le P_t]} \dif \nu(x)
=
    (\star)
\end{align*}
For each $x \in [0,1]$ and for any $s \in \N$, let $B_s(x)\coloneq \I\{x \le V_s\}$, and notice that $B_1(x),B_2(x),\dots$ is an i.i.d.\ sequence of Bernoulli random variables of parameter $1-F(x)$. 
Let $V^{(1)}_{2(t-1)},\dots,V^{(2(t-1))}_{2(t-1)}$ be the order statistics of the observed sample $V_1,\dots,V_{2(t-1)}$.
For any $x \in [m,1]$, observing that $F(x)-\frac{1}{2} \ge 0$ and $\Pb[x \le P_t] \le \Pb\lsb{x \le V_{2(t-1)}^{(t)}} \le \Pb\lsb{\sum_{s=1}^{2(t-1)} B_s(x) \ge t-1}$, we can leverage Hoeffding's inequality to obtain
\begin{align*}
&
    \Pb[x \le P_t]
\le
    \Pb\lsb{\sum_{s=1}^{2(t-1)} B_s(x) \ge t-1}
=
    \Pb\lsb{\sum_{s=1}^{2(t-1)} \frac{B_s(x)}{2(t-1)}  - (1-F(x)) \ge \frac{t-1}{2(t-1)} - (1-F(x))}
\\
&\qquad=
    \Pb\lsb{\sum_{s=1}^{2(t-1)} \frac{B_s(x)}{2(t-1)} - (1-F(x)) \ge F(x)-\frac{1}{2}}
\le
    e^{-4(t-1)\lrb{F(x)-\frac{1}{2}}^2}
=
    e^{-4(t-1)\lrb{\nu[0,x]-\frac{1}{2}}^2}\;,
\end{align*}
from which, by the change of variable formula \cite[Proposition 4.10, Chapter 1]{revuz2013continuous}, it follows also that
\begin{align*}
    (\star)
&
\le
    \int_{[m,1]} \sqrt{\exp\lrb{-4(t-1)\lrb{\nu\bsb{[0,x]}-\frac{1}{2}}^2}} \dif \nu(x)
=
    \int_{1/2}^1 \exp\lrb{-2(t-1)\lrb{\frac{1}{2}-u}^2} \dif u
\\
&
\le
    \frac{1}{\sqrt{2(t-1)}}\int_{0}^\infty \exp\lrb{-r^2} \dif r
=
    \frac{\sqrt{\pi}}{2\sqrt{2}} \cdot \frac{1}{\sqrt{t-1}} \;,
\end{align*}
and hence $(I) \le \frac{\pi}{8(t-1)}$.
Analogously, we can prove that $(II) \le \frac{\pi}{8(t-1)}$.
Hence,
\[
    R_T
\le
    \frac{1}{2} + \frac{\pi}{2} \cdot \sum_{t=2}^{T} \frac{1}{t-1}
=
    \frac{1}{2} + \frac{\pi}{2} + \frac{\pi}{2} \cdot \sum_{t=2}^{T-1} \int_{t-1}^{t}\frac{1}{t} \dif s
\le
    \frac{1}{2} + \frac{\pi}{2} + \frac{\pi}{2} \cdot \int_{1}^{T-1}\frac{1}{s} \dif s
=
    \frac{1}{2} + \frac{\pi}{2} \brb{1 + \ln (T-1)}\;.  \qedhere
\]

\end{proof}

We now establish the optimality of \fem{} by demonstrating a matching $\Omega(\ln T)$ regret lower bound.
We remark that this result holds even when competing against underlying distributions that have a $2$-Lipschitz cdf, thus proving the optimality of \fem{} even under the Lipschitz cdf assumption.
\begin{theorem}
    \label{t:lower-bound-full-M}
    There exist two numerical constants $c_1$ and $c_2$ such that, for any time horizon $T \ge c_2$, the worst-case regret of any full-feedback algorithm satisfies
    \[
        \sup_{\nu\in\cD_2} R_T^\nu
    \ge
       c_1 \ln T \;,
    \]
    where $R_T^\nu$ is the regret at time $T$ of the algorithm when the i.i.d.\ sequence of traders' valuations follows the distribution $\nu$, and $\cD_2$ is the set of all distributions $\nu$ that admit a $2$-Lipschitz cdf.
\end{theorem}
Due to space constraints, we defer the (long and technical) proof of this result to \Cref{s:appe-lower-bound-full-M} and only present a short, high-level sketch here.
\begin{figure*}
    \def\scaleFactor{7/8}
    \def\scaleFactorOne{1/10}
    \def\scaleFactorTwo{2.5/5}
    \centering
    \begin{tikzpicture}[scale=4.2]

    \draw[lightgray, very thin] 
        (1/8,0.6 * \scaleFactor) -- (1/8,0)
        (7/8,0.6 * \scaleFactor) -- (7/8,0)
        (1,0.31 * \scaleFactor) -- (1,0)
        (0, 0.6 * \scaleFactor) -- (1/8,0.6 * \scaleFactor)
        (0, 0.31 * \scaleFactor) -- (7/8,0.31 * \scaleFactor)
    ;

    \draw[lightgray, very thin] 
        (-0.16,0.14 * \scaleFactorTwo) -- (0,0.14 * \scaleFactorTwo)
        (-0.16, 0.31 * \scaleFactorTwo) -- (7/8,0.31 * \scaleFactorTwo)
    ;

    \draw (0, 0.14 * \scaleFactor) -- (1/8, 0.14 * \scaleFactor);
    \draw (7/8, 0.31 * \scaleFactor) -- (1, 0.31 * \scaleFactor);
    \draw (1/8, 0.6 * \scaleFactor) -- (7/8, 0.6 * \scaleFactor);

    \draw[red] (0, 0.14 * \scaleFactorTwo) -- (1/8, 0.14 * \scaleFactorTwo);
    \draw[red] (7/8, 0.31 * \scaleFactorTwo) -- (1, 0.31 * \scaleFactorTwo);

    \draw (0,0) -- (0,-0.02) node[below]{$0$}
        (1/8,0) -- (1/8,-0.02) node[below]{$\nicefrac18$}
        (7/8,0) -- (7/8,-0.02) node[below]{$\nicefrac78$}
        (1,0) -- (1,-0.02) node[below]{$1$}
        (0,0.14 * \scaleFactor) -- (-0.02,0.14 * \scaleFactor) node[left]{$2\e$}
        (0,0.31 * \scaleFactor) -- (-0.02,0.31 * \scaleFactor) node[left]{$2(1-\e)$}
        (0,0.6 * \scaleFactor) -- (-0.02,0.6 * \scaleFactor) node[left]{$1$}
    ;

    \draw (0,0.14 * \scaleFactorTwo) -- (-0.02,0.14 * \scaleFactorTwo)
        (0,0.31 * \scaleFactorTwo) -- (-0.02,0.31 * \scaleFactorTwo)
    ;
    \draw
        (-0.15,0.14 * \scaleFactorTwo) node[left]{\textcolor{red}{$2\e'$}}
        (-0.15,0.31 * \scaleFactorTwo) node[left]{\textcolor{red}{$2(1-\e')$}}
    ;
    
    \draw[<->] (0,0.6) -- (0,0) -- (1.1,0);

    \end{tikzpicture}
    \hspace{16pt}
    \begin{tikzpicture}[scale=4.2]

    \draw[lightgray, very thin] 
        (1/8,
            {
                2 *
                (
                    2 * 0.31 * \scaleFactor * 1/8 
                )
                *
                (
                    1
                    -
                    (
                        2 * 0.31 * \scaleFactor * 1/8 
                    )
                )
            }
        ) 
        -- (1/8,0)
        (7/8,
            { 
                2 *
                (
                    2 * 0.31 * \scaleFactorOne - 1 - 2 * ( 0.31 * \scaleFactorOne - 1 ) * 7/8
                )
                *
                (
                    1
                    -
                    (
                        2 * 0.31 * \scaleFactorOne - 1 - 2 * ( 0.31 * \scaleFactorOne - 1 ) * 7/8
                    )
                )
            }
        ) -- (7/8,0)
        ( { ( 5 - 2 * 0.31 * \scaleFactor ) / 8 },
            { 
                2 *
                (
                    2 * 0.31 * \scaleFactor * ( ( 5 - 2 * 0.31 * \scaleFactor ) / 8 ) * ( ( ( 5 - 2 * 0.31 * \scaleFactor ) / 8 ) <= 1/8 )
                    +
                    ( (2 * 0.31 * \scaleFactor - 1)/8 + ( ( 5 - 2 * 0.31 * \scaleFactor ) / 8 ) ) * ( 1/8 < ( ( 5 - 2 * 0.31 * \scaleFactor ) / 8 ) ) * ( ( ( 5 - 2 * 0.31 * \scaleFactor ) / 8 ) < 7/8 )
                    +
                    ( 2 * 0.31 * \scaleFactor - 1 - 2 * ( 0.31 * \scaleFactor - 1 ) * ( ( 5 - 2 * 0.31 * \scaleFactor ) / 8 ) ) * ( ( ( 5 - 2 * 0.31 * \scaleFactor ) / 8 ) >= 7/8 )
                )
                *
                (
                    1
                    -
                    (
                        2 * 0.31 * \scaleFactor * ( ( 5 - 2 * 0.31 * \scaleFactor ) / 8 ) * ( ( ( 5 - 2 * 0.31 * \scaleFactor ) / 8 ) <= 1/8 )
                        +
                        ( (2 * 0.31 * \scaleFactor - 1)/8 + ( ( 5 - 2 * 0.31 * \scaleFactor ) / 8 ) ) * ( 1/8 < ( ( 5 - 2 * 0.31 * \scaleFactor ) / 8 ) ) * ( ( ( 5 - 2 * 0.31 * \scaleFactor ) / 8 ) < 7/8 )
                        +
                        ( 2 * 0.31 * \scaleFactor - 1 - 2 * ( 0.31 * \scaleFactor - 1 ) * ( ( 5 - 2 * 0.31 * \scaleFactor ) / 8 ) ) * ( ( ( 5 - 2 * 0.31 * \scaleFactor ) / 8 ) >= 7/8 )
                    )
                )
            }
        )
        -- ( { ( 5 - 2 * 0.31 * \scaleFactor ) / 8 }, 0 )
        ( { ( 5 - 2 * 0.31 * \scaleFactorOne ) / 8 },
            { 
                2 *
                (
                    2 * 0.31 * \scaleFactorOne * ( ( 5 - 2 * 0.31 * \scaleFactorOne ) / 8 ) * ( ( ( 5 - 2 * 0.31 * \scaleFactorOne ) / 8 ) <= 1/8 )
                    +
                    ( (2 * 0.31 * \scaleFactorOne - 1)/8 + ( ( 5 - 2 * 0.31 * \scaleFactorOne ) / 8 ) ) * ( 1/8 < ( ( 5 - 2 * 0.31 * \scaleFactorOne ) / 8 ) ) * ( ( ( 5 - 2 * 0.31 * \scaleFactorOne ) / 8 ) < 7/8 )
                    +
                    ( 2 * 0.31 * \scaleFactorOne - 1 - 2 * ( 0.31 * \scaleFactorOne - 1 ) * ( ( 5 - 2 * 0.31 * \scaleFactorOne ) / 8 ) ) * ( ( ( 5 - 2 * 0.31 * \scaleFactorOne ) / 8 ) >= 7/8 )
                )
                *
                (
                    1
                    -
                    (
                        2 * 0.31 * \scaleFactorOne * ( ( 5 - 2 * 0.31 * \scaleFactorOne ) / 8 ) * ( ( ( 5 - 2 * 0.31 * \scaleFactorOne ) / 8 ) <= 1/8 )
                        +
                        ( (2 * 0.31 * \scaleFactorOne - 1)/8 + ( ( 5 - 2 * 0.31 * \scaleFactorOne ) / 8 ) ) * ( 1/8 < ( ( 5 - 2 * 0.31 * \scaleFactorOne ) / 8 ) ) * ( ( ( 5 - 2 * 0.31 * \scaleFactorOne ) / 8 ) < 7/8 )
                        +
                        ( 2 * 0.31 * \scaleFactorOne - 1 - 2 * ( 0.31 * \scaleFactorOne - 1 ) * ( ( 5 - 2 * 0.31 * \scaleFactorOne ) / 8 ) ) * ( ( ( 5 - 2 * 0.31 * \scaleFactorOne ) / 8 ) >= 7/8 )
                    )
                )
            }
        )
        -- ( { ( 5 - 2 * 0.31 * \scaleFactorOne ) / 8 }, 0 )
        ( { ( 5 - 2 * 0.31 * \scaleFactor ) / 8 }, -0.01 ) -- ( { ( 5 - 2 * 0.31 * \scaleFactor ) / 8 }, -0.02 ) -- ( { ( 5 - 2 * 0.31 * \scaleFactorOne ) / 8 }, -0.02 ) -- ( { ( 5 - 2 * 0.31 * \scaleFactorOne ) / 8 }, -0.01 )
        ( { ( 5 - 2 * 0.31 * \scaleFactorOne ) / 8 + 0.01 }, 
        {
        2 *
                (
                    2 * 0.31 * \scaleFactorOne * ( ( 5 - 2 * 0.31 * \scaleFactorOne ) / 8 ) * ( ( ( 5 - 2 * 0.31 * \scaleFactorOne ) / 8 ) <= 1/8 )
                    +
                    ( (2 * 0.31 * \scaleFactorOne - 1)/8 + ( ( 5 - 2 * 0.31 * \scaleFactorOne ) / 8 ) ) * ( 1/8 < ( ( 5 - 2 * 0.31 * \scaleFactorOne ) / 8 ) ) * ( ( ( 5 - 2 * 0.31 * \scaleFactorOne ) / 8 ) < 7/8 )
                    +
                    ( 2 * 0.31 * \scaleFactorOne - 1 - 2 * ( 0.31 * \scaleFactorOne - 1 ) * ( ( 5 - 2 * 0.31 * \scaleFactorOne ) / 8 ) ) * ( ( ( 5 - 2 * 0.31 * \scaleFactorOne ) / 8 ) >= 7/8 )
                )
                *
                (
                    1
                    -
                    (
                        2 * 0.31 * \scaleFactorOne * ( ( 5 - 2 * 0.31 * \scaleFactorOne ) / 8 ) * ( ( ( 5 - 2 * 0.31 * \scaleFactorOne ) / 8 ) <= 1/8 )
                        +
                        ( (2 * 0.31 * \scaleFactorOne - 1)/8 + ( ( 5 - 2 * 0.31 * \scaleFactorOne ) / 8 ) ) * ( 1/8 < ( ( 5 - 2 * 0.31 * \scaleFactorOne ) / 8 ) ) * ( ( ( 5 - 2 * 0.31 * \scaleFactorOne ) / 8 ) < 7/8 )
                        +
                        ( 2 * 0.31 * \scaleFactorOne - 1 - 2 * ( 0.31 * \scaleFactorOne - 1 ) * ( ( 5 - 2 * 0.31 * \scaleFactorOne ) / 8 ) ) * ( ( ( 5 - 2 * 0.31 * \scaleFactorOne ) / 8 ) >= 7/8 )
                    )
                )
                }
            )
        -- (1.02, { 2 *
                (
                    2 * 0.31 * \scaleFactorOne * ( ( 5 - 2 * 0.31 * \scaleFactorOne ) / 8 ) * ( ( ( 5 - 2 * 0.31 * \scaleFactorOne ) / 8 ) <= 1/8 )
                    +
                    ( (2 * 0.31 * \scaleFactorOne - 1)/8 + ( ( 5 - 2 * 0.31 * \scaleFactorOne ) / 8 ) ) * ( 1/8 < ( ( 5 - 2 * 0.31 * \scaleFactorOne ) / 8 ) ) * ( ( ( 5 - 2 * 0.31 * \scaleFactorOne ) / 8 ) < 7/8 )
                    +
                    ( 2 * 0.31 * \scaleFactorOne - 1 - 2 * ( 0.31 * \scaleFactorOne - 1 ) * ( ( 5 - 2 * 0.31 * \scaleFactorOne ) / 8 ) ) * ( ( ( 5 - 2 * 0.31 * \scaleFactorOne ) / 8 ) >= 7/8 )
                )
                *
                (
                    1
                    -
                    (
                        2 * 0.31 * \scaleFactorOne * ( ( 5 - 2 * 0.31 * \scaleFactorOne ) / 8 ) * ( ( ( 5 - 2 * 0.31 * \scaleFactorOne ) / 8 ) <= 1/8 )
                        +
                        ( (2 * 0.31 * \scaleFactorOne - 1)/8 + ( ( 5 - 2 * 0.31 * \scaleFactorOne ) / 8 ) ) * ( 1/8 < ( ( 5 - 2 * 0.31 * \scaleFactorOne ) / 8 ) ) * ( ( ( 5 - 2 * 0.31 * \scaleFactorOne ) / 8 ) < 7/8 )
                        +
                        ( 2 * 0.31 * \scaleFactorOne - 1 - 2 * ( 0.31 * \scaleFactorOne - 1 ) * ( ( 5 - 2 * 0.31 * \scaleFactorOne ) / 8 ) ) * ( ( ( 5 - 2 * 0.31 * \scaleFactorOne ) / 8 ) >= 7/8 )
                    )
                )
                }
            ) -- (1.02, 
            { 
                2 *
                (
                    2 * 0.31 * \scaleFactor * ( ( 5 - 2 * 0.31 * \scaleFactorOne ) / 8 ) * ( ( ( 5 - 2 * 0.31 * \scaleFactorOne ) / 8 ) <= 1/8 )
                    +
                    ( (2 * 0.31 * \scaleFactor - 1)/8 + ( ( 5 - 2 * 0.31 * \scaleFactorOne ) / 8 ) ) * ( 1/8 < ( ( 5 - 2 * 0.31 * \scaleFactorOne ) / 8 ) ) * ( ( ( 5 - 2 * 0.31 * \scaleFactorOne ) / 8 ) < 7/8 )
                    +
                    ( 2 * 0.31 * \scaleFactor - 1 - 2 * ( 0.31 * \scaleFactor - 1 ) * ( ( 5 - 2 * 0.31 * \scaleFactorOne ) / 8 ) ) * ( ( ( 5 - 2 * 0.31 * \scaleFactorOne ) / 8 ) >= 7/8 )
                )
                *
                (
                    1
                    -
                    (
                        2 * 0.31 * \scaleFactor * ( ( 5 - 2 * 0.31 * \scaleFactorOne ) / 8 ) * ( ( ( 5 - 2 * 0.31 * \scaleFactorOne ) / 8 ) <= 1/8 )
                        +
                        ( (2 * 0.31 * \scaleFactor - 1)/8 + ( ( 5 - 2 * 0.31 * \scaleFactorOne ) / 8 ) ) * ( 1/8 < ( ( 5 - 2 * 0.31 * \scaleFactorOne ) / 8 ) ) * ( ( ( 5 - 2 * 0.31 * \scaleFactorOne ) / 8 ) < 7/8 )
                        +
                        ( 2 * 0.31 * \scaleFactor - 1 - 2 * ( 0.31 * \scaleFactor - 1 ) * ( ( 5 - 2 * 0.31 * \scaleFactorOne ) / 8 ) ) * ( ( ( 5 - 2 * 0.31 * \scaleFactorOne ) / 8 ) >= 7/8 )
                    )
                )
            }
            ) -- ( { ( 5 - 2 * 0.31 * \scaleFactorOne ) / 8 + 0.01 },
            { 
                2 *
                (
                    2 * 0.31 * \scaleFactor * ( ( 5 - 2 * 0.31 * \scaleFactorOne ) / 8 ) * ( ( ( 5 - 2 * 0.31 * \scaleFactorOne ) / 8 ) <= 1/8 )
                    +
                    ( (2 * 0.31 * \scaleFactor - 1)/8 + ( ( 5 - 2 * 0.31 * \scaleFactorOne ) / 8 ) ) * ( 1/8 < ( ( 5 - 2 * 0.31 * \scaleFactorOne ) / 8 ) ) * ( ( ( 5 - 2 * 0.31 * \scaleFactorOne ) / 8 ) < 7/8 )
                    +
                    ( 2 * 0.31 * \scaleFactor - 1 - 2 * ( 0.31 * \scaleFactor - 1 ) * ( ( 5 - 2 * 0.31 * \scaleFactorOne ) / 8 ) ) * ( ( ( 5 - 2 * 0.31 * \scaleFactorOne ) / 8 ) >= 7/8 )
                )
                *
                (
                    1
                    -
                    (
                        2 * 0.31 * \scaleFactor * ( ( 5 - 2 * 0.31 * \scaleFactorOne ) / 8 ) * ( ( ( 5 - 2 * 0.31 * \scaleFactorOne ) / 8 ) <= 1/8 )
                        +
                        ( (2 * 0.31 * \scaleFactor - 1)/8 + ( ( 5 - 2 * 0.31 * \scaleFactorOne ) / 8 ) ) * ( 1/8 < ( ( 5 - 2 * 0.31 * \scaleFactorOne ) / 8 ) ) * ( ( ( 5 - 2 * 0.31 * \scaleFactorOne ) / 8 ) < 7/8 )
                        +
                        ( 2 * 0.31 * \scaleFactor - 1 - 2 * ( 0.31 * \scaleFactor - 1 ) * ( ( 5 - 2 * 0.31 * \scaleFactorOne ) / 8 ) ) * ( ( ( 5 - 2 * 0.31 * \scaleFactorOne ) / 8 ) >= 7/8 )
                    )
                )
            }
        )
    ;

    \draw[domain = 0:1, samples = 350] plot (\x, 
    { 
        2 *
        (
            2 * 0.31 * \scaleFactor * \x * ( \x <= 1/8 )
            +
            ( (2 * 0.31 * \scaleFactor - 1)/8 + \x ) * ( 1/8 < \x ) * ( \x < 7/8 )
            +
            ( 2 * 0.31 * \scaleFactor - 1 - 2 * ( 0.31 * \scaleFactor - 1 ) * \x ) * ( \x >= 7/8 )
        )
        *
        (
            1
            -
            (
                2 * 0.31 * \scaleFactor * \x * ( \x <= 1/8 )
                +
                ( (2 * 0.31 * \scaleFactor - 1)/8 + \x ) * ( 1/8 < \x ) * ( \x < 7/8 )
                +
                ( 2 * 0.31 * \scaleFactor - 1 - 2 * ( 0.31 * \scaleFactor - 1 ) * \x ) * ( \x >= 7/8 )
            )
        )
    });

   \draw[domain = 0:1, samples = 350, red] plot (\x, 
    { 
        2 *
        (
            2 * 0.31 * \scaleFactorOne * \x * ( \x <= 1/8 )
            +
            ( (2 * 0.31 * \scaleFactorOne - 1)/8 + \x ) * ( 1/8 < \x ) * ( \x < 7/8 )
            +
            ( 2 * 0.31 * \scaleFactorOne - 1 - 2 * ( 0.31 * \scaleFactorOne - 1 ) * \x ) * ( \x >= 7/8 )
        )
        *
        (
            1
            -
            (
                2 * 0.31 * \scaleFactorOne * \x * ( \x <= 1/8 )
                +
                ( (2 * 0.31 * \scaleFactorOne - 1)/8 + \x ) * ( 1/8 < \x ) * ( \x < 7/8 )
                +
                ( 2 * 0.31 * \scaleFactorOne - 1 - 2 * ( 0.31 * \scaleFactorOne - 1 ) * \x ) * ( \x >= 7/8 )
            )
        )
    });

    \draw (0,0) -- (0,-0.02) node[below]{$0$}
        (1/8,0) -- (1/8,-0.02) node[below]{$\nicefrac18$}
        (7/8,0) -- (7/8,-0.02) node[below]{$\nicefrac78$}
        (1,0) -- (1,-0.02) node[below]{$1$}
        ( { ( ( 5 - 2 * 0.31 * \scaleFactor ) / 8 ) * 0.5  + ( ( 5 - 2 * 0.31 * \scaleFactorOne ) / 8 ) * 0.5}, -0.02 ) node[below]{$\Theta( \labs{ \e-\e' } )$}
        (1.02, 
        {
            0.5 *
            2 *
                (
                    2 * 0.31 * \scaleFactor * ( ( 5 - 2 * 0.31 * \scaleFactorOne ) / 8 ) * ( ( ( 5 - 2 * 0.31 * \scaleFactorOne ) / 8 ) <= 1/8 )
                    +
                    ( (2 * 0.31 * \scaleFactor - 1)/8 + ( ( 5 - 2 * 0.31 * \scaleFactorOne ) / 8 ) ) * ( 1/8 < ( ( 5 - 2 * 0.31 * \scaleFactorOne ) / 8 ) ) * ( ( ( 5 - 2 * 0.31 * \scaleFactorOne ) / 8 ) < 7/8 )
                    +
                    ( 2 * 0.31 * \scaleFactor - 1 - 2 * ( 0.31 * \scaleFactor - 1 ) * ( ( 5 - 2 * 0.31 * \scaleFactorOne ) / 8 ) ) * ( ( ( 5 - 2 * 0.31 * \scaleFactorOne ) / 8 ) >= 7/8 )
                )
                *
                (
                    1
                    -
                    (
                        2 * 0.31 * \scaleFactor * ( ( 5 - 2 * 0.31 * \scaleFactorOne ) / 8 ) * ( ( ( 5 - 2 * 0.31 * \scaleFactorOne ) / 8 ) <= 1/8 )
                        +
                        ( (2 * 0.31 * \scaleFactor - 1)/8 + ( ( 5 - 2 * 0.31 * \scaleFactorOne ) / 8 ) ) * ( 1/8 < ( ( 5 - 2 * 0.31 * \scaleFactorOne ) / 8 ) ) * ( ( ( 5 - 2 * 0.31 * \scaleFactorOne ) / 8 ) < 7/8 )
                        +
                        ( 2 * 0.31 * \scaleFactor - 1 - 2 * ( 0.31 * \scaleFactor - 1 ) * ( ( 5 - 2 * 0.31 * \scaleFactorOne ) / 8 ) ) * ( ( ( 5 - 2 * 0.31 * \scaleFactorOne ) / 8 ) >= 7/8 )
                    )
                )
                +
                0.5 *
                2 *
                (
                    2 * 0.31 * \scaleFactorOne * ( ( 5 - 2 * 0.31 * \scaleFactorOne ) / 8 ) * ( ( ( 5 - 2 * 0.31 * \scaleFactorOne ) / 8 ) <= 1/8 )
                    +
                    ( (2 * 0.31 * \scaleFactorOne - 1)/8 + ( ( 5 - 2 * 0.31 * \scaleFactorOne ) / 8 ) ) * ( 1/8 < ( ( 5 - 2 * 0.31 * \scaleFactorOne ) / 8 ) ) * ( ( ( 5 - 2 * 0.31 * \scaleFactorOne ) / 8 ) < 7/8 )
                    +
                    ( 2 * 0.31 * \scaleFactorOne - 1 - 2 * ( 0.31 * \scaleFactorOne - 1 ) * ( ( 5 - 2 * 0.31 * \scaleFactorOne ) / 8 ) ) * ( ( ( 5 - 2 * 0.31 * \scaleFactorOne ) / 8 ) >= 7/8 )
                )
                *
                (
                    1
                    -
                    (
                        2 * 0.31 * \scaleFactorOne * ( ( 5 - 2 * 0.31 * \scaleFactorOne ) / 8 ) * ( ( ( 5 - 2 * 0.31 * \scaleFactorOne ) / 8 ) <= 1/8 )
                        +
                        ( (2 * 0.31 * \scaleFactorOne - 1)/8 + ( ( 5 - 2 * 0.31 * \scaleFactorOne ) / 8 ) ) * ( 1/8 < ( ( 5 - 2 * 0.31 * \scaleFactorOne ) / 8 ) ) * ( ( ( 5 - 2 * 0.31 * \scaleFactorOne ) / 8 ) < 7/8 )
                        +
                        ( 2 * 0.31 * \scaleFactorOne - 1 - 2 * ( 0.31 * \scaleFactorOne - 1 ) * ( ( 5 - 2 * 0.31 * \scaleFactorOne ) / 8 ) ) * ( ( ( 5 - 2 * 0.31 * \scaleFactorOne ) / 8 ) >= 7/8 )
                    )
                )
        }
        ) node[right]{$\Theta \brb{ \labs{\e - \e'}^2}$}
    ;

    \draw[<->] (0,0.6) -- (0,0) -- (1.1,0);

    \end{tikzpicture}
    \caption{Qualitative plots of the densities $f_\e$, \textcolor{red}{$f_{\e'}$} (left) and corresponding expected rewards (right) used in the proof of \Cref{t:lower-bound-full-M} for two values $\e, \textcolor{red}{\e'} >0$. Note that the difference in reward by posting a price that is optimal for one instance $\e'$ when the actual instance is $\e$ is $\Theta\brb{\labs{\e-\e'}^2}$.   
    \label{f:lower-bound-density}}
\end{figure*}
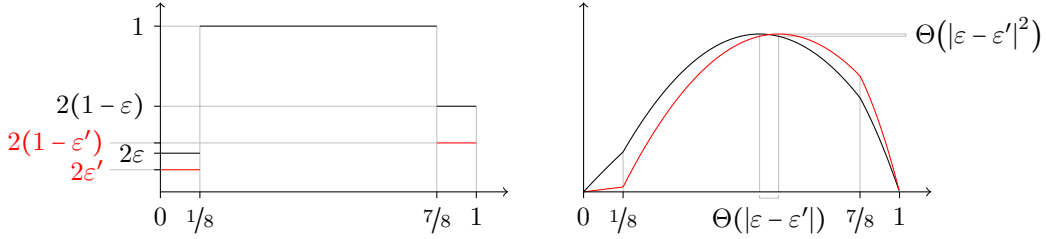
\begin{proof}[Proof sketch.]
In the proof, we build a family of $2$-Lipschitz cdfs $F_\e$ parameterized by $\e \in [0,1]$, so that if two instances are parameterized by $\e$ and $\e'$ respectively, then their medians are $\Theta\brb{\labs{\e-\e'}}$-away from each other (\cref{f:lower-bound-density}).
The high-level idea is to leverage a Bayesian argument to show that if the underlying instance $F_E$ is such that $E$ is drawn uniformly at random in $[0,1]$, then, at round $t$, the broker cannot reliably determine prices that are much closer than $1/\sqrt{t}$ to the corresponding median $m_E$ when distances are measured with respect to the metric induced by the cdf $F_E$.
This, together with our key \Cref{l:mysticus}, leads to the conclusion.
\end{proof}

\section{2-Bit Feedback}
We start the study of the $2$-bit feedback case under the assumption that the traders' valuation distribution admits a Lipschitz cdf $F$.
The algorithm we propose (\cref{a:MBS}) is based on the following observation: by posting any price $p$, the broker has access to two noisy realizations of $F(p)$.
Recalling that \cref{l:mysticus} suggests tracking the median of $F$ (i.e., a point $m$ at which $F(m) = 1/2$), and since $F$ is a non-decreasing function, we can proceed using a natural binary search strategy to move toward the median.
This can be done in epochs:  in each one, we repeatedly test a (dyadic) price until the first time we can confidently decide that the median is to the left or right of the current price. 
\begin{algorithm}
\textbf{Input:} Confidence parameter $\delta \in (0,1) $, time horizon $n\in \N$\; 
\textbf{Initialization:} $Q_1 \coloneqq \frac{1}{2}$, $\tau \coloneqq 1$, $t \coloneqq 1$\;
\While
{%
    time $t \le n$
}
{
    Let $s \coloneqq 0$, $Y_{\tau,s} \coloneqq 0$, $t_{\tau-1}\coloneqq t-1$\;
    \While
    {
        time $t \le n$
    }
    {
        Post $P_t \coloneqq Q_{\tau}$ and 
        receive feedback $\I\{V_{2t-1} \le P_t\}$, $\I\{V_{2t} \le P_t\}$\;
        Update  $s \coloneqq s + 2$, $Y_{\tau,s} \coloneqq Y_{\tau,s-2} + \I\{V_{2t-1} \le P_t\} + \I\{V_{2t} \le P_t\}$, $t \coloneqq t + 1$\;
        \textbf{if}
            $\frac{1}{s}Y_{\tau,s} + \sqrt{\frac{\ln(2/\delta)}{2s}} < \frac{1}{2}$
        \textbf{then} 
        let $Q_{\tau + 1} \coloneqq Q_{\tau + 1} + \frac{1}{2^{\tau + 1}}$, $s_\tau \coloneqq s$, 
        $\tau \coloneqq \tau + 1$, and
        \textbf{break}\;
        \textbf{else if} 
            $\frac{1}{s}Y_{\tau,s} - \sqrt{\frac{\ln(2/\delta)}{2s}}> \frac{1}{2} $
        \textbf{then}
            let $Q_{\tau + 1} \coloneqq Q_{\tau + 1} - \frac{1}{2^{\tau + 1}}$, $s_\tau \coloneqq s$,
            $\tau \coloneqq \tau + 1$, 
            and
            \textbf{break}\;
    }
}
\caption{Median Binary Search (MBS)}
\label{a:MBS}
\end{algorithm}

We now show that a suitably tuned \Cref{a:MBS} has regret guarantees of $O\brb{\ln(MT)\ln(T)}$. 

\begin{theorem}
\label{t:mbs}
If $\nu$ has an $M$-Lipschitz cdf $F$, for some $M\ge 1$, then, for all time horizons $T\in \N$, the regret of MBS tuned with parameters $\delta\coloneqq \fracc{2}{T^3}$ and $n\coloneqq T$ satisfies
\[
    R_T
\le
    2 + 6\log_2(MT) \ln(T)\;.
\]
\end{theorem}

Due to space constraints, we defer the full proof of this result to \cref{s:appe:t:mbs}.

\begin{proof}[Proof sketch]
The proof is based on the following observations.
First, during an epoch where a price $p$ is tested, given that one has to distinguish if the parameter $F(p)$ of a sequence of Bernoulli random variables is bigger or smaller than $1/2$, a concentration argument shows that the duration of this epoch is at most $O\brb{\ln(1/\delta)/\brb{1/2-F(p)}^2}$, where $\delta$ is the confidence parameter.
Second, by \cref{l:mysticus}, the broker regrets $2 \brb{1/2-F(p)}^2$ by playing a price $p$, and hence the total regret of an epoch where the broker tests $p$ is at most $O\brb{\ln (1/\delta)}$.
We then use the fact that the $F$ is $M$-Lipschitz to prove that after at most $\log_2(MT)$ epochs, the cumulative regret that the algorithm suffers from that point onward is constant, and conclude by showing that a proper tuning of the confidence parameter $\delta$ (that does not require the knowledge of $M$) leads to the stated guarantees.
\end{proof}

We now show that \cref{a:MBS} is optimal, up to a multiplicative $\ln T$ term.
\begin{theorem}
    \label{t:lower-bound-2-bit-M}
    There exist two numerical constants $c_1$ and $c_2$ such that for any $M \ge 16$ and any time horizon $T \ge c_2\log_2(M)$, the worst-case regret of any $2$-bit feedback algorithm satisfies
    \[
        \sup_{\nu\in\cD_M} R_T^\nu
    \ge
       c_1 \ln (MT) \;,
    \]
    where $R_T^\nu$ is the regret at time $T$ of the algorithm when the i.i.d.\ sequence of traders' valuations follows the distribution $\nu$, and $\cD_M$ is the set of all distributions $\nu$ that admits an $M$-Lipschitz cdf.
\end{theorem}

Due to space constraints, we defer the full proof of this result to \Cref{s:appe:t:lower-bound-2-bit-M}.

\begin{proof}[Proof sketch]
The proof builds a family of densities, each concentrated in a different region of length $1/M$. 
To avoid suffering linear regret, the broker has to detect this region. 
To accomplish this task, we show that the broker is essentially forced to solve a binary search problem that needs $\log_2(M)$ rounds in each of which the instantaneous regret is constant. 
Noticing that any regret lower bound for full-feedback algorithms also applies to $2$-bit feedback algorithms, the $\ln T$ lower bound of \Cref{t:lower-bound-full-M} together with the binary search $\ln M$ lower bound yield a lower bound of $\Omega\brb{\max(\ln T ,\ln M)} = \Omega\brb{ \ln(MT)}$.
\end{proof}

\section{Non-Lipschitz or Discontinuous Pdfs}
\label{s:beyond-regularity}

We now investigate how the problem changes if we lift the assumption that $\nu$ has a Lipschitz or continuous cdf.
First, note that when the cdf of $\nu$ is not continuous, \Cref{l:mysticus}, and, consequently, the guarantees of \Cref{t:ftm}, no longer hold.
Indeed, in general, no full-feedback algorithm can achieve regret guarantees better than $\sqrt{T}$.
As shown in the proof of the next theorem, the reason is that our problem contains instances that resemble online learning with expert advice (with 2 experts), which has a known lower bound of $\Omega(\sqrt{T})$.
\begin{theorem}
\label{t:lower-full-info-sqrtT}
There exist two numerical constants $c_1$ and $c_2$ such that, for any time horizon $T \ge c_2$, the worst-case regret of any full feedback algorithm satisfies
\[
    \sup_{\nu\in\cD} R_T^\nu
\ge
   c_1 \sqrt{T} \;,
\]
where $R_T^\nu$ is the regret at time $T$ of the algorithm when the i.i.d.\ sequence of traders' valuations follows the distribution $\nu$, and $\cD$ is the set of all distributions $\nu$.
\end{theorem}
\begin{proof}[Proof sketch]
For each $\e \in [-1/4,1/4]$, define
$
    \nu_\e \coloneqq \frac{1-\e}{4} \delta_0 + \frac{1}{4} \delta_{1/3} + \frac{1}{4} \delta_{2/3} + \frac{1+\e}{4}\delta_1
$,
where, for any $a \in \bbR$, we denoted by $\delta_a$ the Dirac's delta probability measure centered at $a$.
Let $(V_{\e,t})_{\e\in[-1/4,1/4], t\in \N}$ be an independent family such that for each $\e \in [-1/4,1/4]$ the sequence $V_{\e,1},V_{\e,2},\dots$ is i.i.d.\ with common distribution $\nu_{\e}$.
For each $\e \in [-1/4,1/4]$, each $t \in \N$, and each $p \in [0,1]$, define $\GFT_{\e,t}(p) \coloneqq \gft(p,V_{\e,2t-1},V_{\e,2t})$.
Straightforward computations show that, for each $\e \in [-1/4,1/4]$ and each $t \in \N$, the function $p \mapsto \E\bsb{\GFT_{\e,t}(p)}$ is maximized at $1/3$ or at $2/3$, with any other point having an expected reward that is less than $31/256$-away from the minimum expected reward achieved at $1/3$ or $2/3$. Furthermore, for any $\e \in [-1/4,1/4]$ and any $t \in \N$, the maximum is at $1/3$ or $2/3$ depending on whether $\e < 0$ or $\e > 0$, given that
$
    \E\bsb{\GFT_{\e,t}(1/3)}
= 
    \frac{11}{16}-\frac{\e}{8}-\frac{\e^2}{8}
$
and 
$
    \E\bsb{\GFT_{\e,t}(2/3)}
= 
    \frac{11}{16}+\frac{\e}{8}-\frac{\e^2}{8}
$,
from which it follows also that
$
    \E\bsb{\GFT_{\e,t}(2/3)} - \E\bsb{\GFT_{\e,t}(1/3)} = \frac{\e}{4}
$.
Hence, in order not to suffer $\Omega(\e T)$ regret, an algorithm has to detect the \emph{sign} of $\e$.
However, a standard information-theoretic argument shows that a sample of order $\Omega(1/\e^2)$ is required in order to detect the sign of $\e$.
During this period, the best any algorithm can do is to play blindly in the set $\{1/3,2/3\}$, incurring in a cumulative regret of order $\Omega\lrb{\frac{1}{\e^2} \cdot \e} = \Omega(1/\e)$.
Overall, any learner has to suffer $\Omega\lrb{ \min\lrb{\frac{1}{\e}, \e T}}$ worst-case regret, which, by tuning $\e = \Theta(1/\sqrt{T})$, leads to a worst-case regret lower bound of $\Omega(\sqrt{T})$.
\end{proof}

We now focus on the upper bound.
A closer look at the proof of \Cref{l:mysticus} shows that if we drop the cdf continuity assumption in the Median Lemma, the formula generalizes to
\[
    \E\bsb{\GFT_t(p)}
=
    2F(p)\brb{1-F(p)} + F(p)F^\circ(p) \eqqcolon \Psi(p)\;, \qquad \forall p \in [0,1], \ \forall t \in \N,
\]
with no assumptions on $\nu$, and where $F$ is the cdf of $\nu$ and we defined $F^\circ(p) \coloneqq \nu \bsb{\{p\}}$.
This suggests the strategy of building an empirical proxy $\hat{\Psi}_t$ of $\Psi$ with the feedback available at time $t$, and posting prices that maximize $\hat{\Psi}_t$.
By replacing the theoretical quantities by their empirical counterparts, for any $t \in \N$ and any $p \in [0,1]$, we can define an empirical proxy for $\Psi(p)$ as follows:
\[
    \hat{\Psi}_{t+1}(p) \coloneqq 2 \frac{1}{2t}\sum_{s=1}^{2t} \I\{V_s \le p\} \frac{1}{2t} \sum_{s=1}^{2t} \I\{p < V_s\} + \frac{1}{2t}\sum_{s=1}^{2t} \I\{V_s \le p\} \frac{1}{2t}\sum_{s=1}^{2t}\I\{V_s = p\} \;.
\]
This definition leads to \Cref{a:ftpsi}.

\begin{algorithm}
Post $P_1 = \frac{1}{2}$ and receive feedback $V_1,V_2$\;
\For
{%
    time $t=2,3,\dots$
}
{
    Post $P_t \in \argmax_{p \in [0,1]} \hat{\Psi}_t(p)$ and receive feedback $V_{2t-1}$, $V_{2t}$\;
}
\caption{Follow the Empirical $\Psi$ (\ftpsi)}
\label{a:ftpsi}
\end{algorithm}

We now state regret guarantees for \Cref{a:ftpsi}.
The proof of the following result (which hinges on showing that $\hat{\Psi}_t$ is \emph{uniformly} close to $\Psi$ with high probability) is deferred to \Cref{s:appe-t:ftpsi}.

\begin{theorem}
\label{t:ftpsi}
    For all time horizons $T\in \N$, the regret of \ftpsi{} satisfies
\[
    R_T
\le
    1 + 8 \sqrt{\pi} \cdot \sqrt{T-1}\;.
\]
\end{theorem}

We conclude by showing that, without the Lipschitz cdf assumption, the $2$-bit feedback problem is unlearnable.
This can be deduced as a simple corollary of the proof of \Cref{t:lower-bound-2-bit-M}.
Specifically, we can obtain a linear worst-case lower bound for any $2$-bit feedback algorithm, even if we assume that the underlying distribution has a continuous cdf.
\begin{theorem}
    \label{t:lower-bound-2-bit-linear}
    There exist two numerical constants $c_1$ and $c_2$ such that, for any time horizon $T \ge c_2$, the worst-case regret of any $2$-bit feedback algorithm satisfies
    \[
        \sup_{\nu\in\cD_c} R_T^\nu
    \ge
       c_1 T \;,
    \]
    where $R_T^\nu$ is the regret at time $T$ of the algorithm when the i.i.d.\ sequence of traders' valuations follows the distribution $\nu$, and $\cD_c$ is the set of all distributions $\nu$ that admits a continuous cdf.
\end{theorem}
\begin{proof}
As a consequence of the last part of the proof of \Cref{t:lower-bound-2-bit-M} (see \cref{s:appe:t:lower-bound-2-bit-M}) we have that, for any time horizon $T \ge 4$, if we set $M \coloneqq  2^T$, then the conditions $M \ge 16$ and $T \ge \log_2(M)$ in that proof holds, and hence, any $2$-bit feedback algorithm has worst-case regret that is at least $\frac{1}{4 \ln 2} \ln M = \frac{1}{4 \ln 2} T$.     
\end{proof}

\section{Conclusions and Open Problems}
Motivated by maximizing trading volume in OTC markets, we proposed a novel objective that departs from the classical \emph{gain-from-trade} reward studied in the bilateral trade literature.
For this new problem, we investigated optimal brokerage strategies from an online learning perspective.
Under the assumption that traders are free to sell or buy depending on the trading price and that traders' valuations form an i.i.d.\ sequence, we provided a complete picture with matching (up to, at most, logarithmic factors) upper and lower bounds in all the proposed settings, fleshing out the role of regularity assumptions in achieving these fast regret rates.

In addition to closing the logarithmic $\ln T$ gap in the regret rate of the $2$-bit feedback setting,
a few other future research directions are to find non-stationary variants of this problem where learning is still achievable, 
investigate trading volume maximization when traders have definite seller and buyer roles, 
and explore the contextual version of the problem when the broker has access to relevant side information before posting each price.

\section*{Acknowledgements}
RC is partially supported the MUR PRIN grant 2022EKNE5K (Learning in Markets and Society), funded by the NextGenerationEU program within the PNRR scheme, the FAIR (Future Artificial Intelligence Research) project, funded by the NextGenerationEU program within the PNRR-PE-AI scheme, the EU Horizon CL4-2022-HUMAN-02 research and innovation action under grant agreement 101120237, project ELIAS (European Lighthouse of AI for Sustainability).
TC gratefully acknowledges the support of the University of Ottawa through grant GR002837 (Start-Up Funds) and that of the Natural Sciences and Engineering Research Council of Canada (NSERC) through grants RGPIN-2023-03688 (Discovery Grants Program) and DGECR2023-00208 (Discovery Grants Program, DGECR - Discovery Launch Supplement).

\bibliographystyle{abbrv}
\bibliography{biblio}


\appendix

\section{Proof of Theorem \ref{t:lower-bound-full-M}}
\label{s:appe-lower-bound-full-M}

For each $\e \in [0,1]$, consider the following density function (see \cref{f:lower-bound-density}, left)
\[
    f_\e \colon [0,1] \to [0,2], \qquad x \mapsto  2\e \I\lcb{x \le \frac{1}{8}} + \I\lcb{\frac{1}{8}<x< \frac{7}{8}} + 2(1-\e) \I\lcb{x \ge \frac{7}{8}} \;,
\]
Notice that, for each $\e \in [0,1]$ the cumulative function associated to the density $f_\e$ is $2$-Lipschitz with explicit expression given by 
\[
    F_\e \colon [0,1] \to [0,1],\quad  x \mapsto 2\e x\I\lcb{x \le \frac{1}{8}} + \lrb{\frac{2\e-1}{8}+x} \I\lcb{\frac{1}{8}<x< \frac{7}{8}} + \brb{2\e - 1 - 2(\e-1)x }\I\lcb{x\ge\frac{7}{8}}\;.
\]
Consider for each $\e \in [0,1]$, an i.i.d.\ sequence $(B_{\e,t})_{t \in \N}$ of Bernoulli random variables of parameter $\e$, an i.i.d.\ sequence $(D_t)_{t\in\N}$ of Bernoulli random variables of parameter $\frac{1}{4}$, an i.i.d.\ sequence $(U_t)_{t \in \N}$ of uniform random variables on $[0,1]$, and a uniform random variable $E$ on $[0,1]$, such that $\lrb{(B_{\e,t})_{t \in \N, \e \in [0,1]} , (D_t)_{t \in \N},(U_t)_{t \in \N}, E}$ is an independent family. 
For each $\e \in [0,1]$ and $t \in \N$, define
    \begin{equation}
        V_{\e,t}
    \coloneqq
        D_t\cdot\lrb{B_{\e,t}\frac{U_t}{8}+(1-B_{\e,t})\frac{7+U_t}{8}} + (1-D_t)\cdot\lrb{\frac{1}{8}+\frac{3}{4}U_t} \;.
    \label{e:representation_of_V}
    \end{equation}
Tedious but straightforward computations show that, for each $\e \in [0,1]$ the sequence $(V_{\e,t})_{t \in \N}$ is i.i.d.\  with common density given by $f_\e$, and this sequence is independent of $E$.
For any $\e\in[0,1]$, $p\in[0,1]$, and $t\in \N$, let $\GFT_{\e,t}(p) \coloneqq \gft(p,V_{\e,2t-1}, V_{\e,2t})$ (for a qualitative representation of its expectation, see \cref{f:lower-bound-density}, right).
We now show how to lower bound the worst-case regret of any arbitrary deterministic algorithm for the full-feedback setting $(\alpha_t)_{t\in \N}$, i.e., a sequence of functions $\alpha_t \colon \brb{[0,1]\times[0,1]}^{t-1} \to [0,1]$ where each element maps past feedback into a price (with the convention that $\alpha_1$ is a number in $[0,1]$).
We remark that we do not lose any generality in considering only deterministic algorithms given that we are in a stochastic i.i.d.\ setting, and the minimax regret over deterministic algorithms coincides with that over randomized algorithms.
For each $t \in \N$, define $\tilde{\alpha}_t \colon \brb{[0,1]\times[0,1]}^{t-1} \to \lsb{\frac{1}{8},\frac{7}{8}}$ equal to $\alpha_t$ whenever $\alpha_t$ takes values in $\lsb{\frac{1}{8},\frac{7}{8}}$, and equal to $1/2$ otherwise.
Notice that for each $\e \in [0,1]$ it holds that $(F_\e \circ \tilde{\alpha}_t) \cdot \lrb{1-F_\e \circ \tilde{\alpha}_t} \ge (F_\e \circ \alpha_t) \cdot \lrb{1-F_\e \circ \alpha_t}$, and hence, due to \cref{l:mysticus}, for each $t \in \N$, it holds that $\E\lsb{ \GFT_{\e,t}\brb{ \tilde{\alpha}_t( V_{\e,1}, \dots, V_{\e,2(t-1)} ) } } \ge \E\lsb{ \GFT_{\e,t}\brb{ \alpha_t( V_{\e,1}, \dots, V_{\e,2(t-1)} ) } }$.
Notice also that for each $\e \in [0,1]$, we have that $m_\e \coloneqq \frac{5-2\e}{8}$ is the unique element in $[0,1]$ such that $F_\e(m_\e)=1/2$. 
For any time horizon $T \ge 144$, we have that the worst-case regret of the algorithm $(\alpha_t)_{t \in \N}$ can be lower bounded as follows
{%
    \small
    \allowdisplaybreaks%
    \begin{align*}
        &\sup_{\nu\in\cD_M} R_T^\nu
    \ge
        \sup_{\e\in[0,1]} 
        \sum_{t=13}^T \E\Bsb{ \GFT_{\e,t}(m_\e) - \GFT_{\e,t}\brb{ \alpha_t( V_{\e,1}, \dots, V_{\e,2(t-1)} ) } }
    \\
    &\quad\ge
        \sup_{\e\in[0,1]} 
        \sum_{t=13}^T \E\Bsb{ \GFT_{\e,t}(m_\e) - \GFT_{\e,t}\brb{ \tilde{\alpha}_t( V_{\e,1}, \dots, V_{\e,2(t-1)} ) } }
    \overset{\spadesuit}=
        \sup_{\e\in[0,1]} 
        \sum_{t=13}^T \E\lsb{ 2\lrb{\frac{1}{2} - F_\e \brb{\tilde{\alpha}_t( V_{\e,1}, \dots, V_{\e,2(t-1)} )}}^2 } 
    \\
    &\quad\overset{\circ}\ge
        \sum_{t=13}^T \E\lsb{ 2\lrb{\frac{1}{2} - F_E \brb{\tilde{\alpha}_t( V_{E,1}, \dots, V_{E,2(t-1)} )}}^2 }
    \overset{\clubsuit}=
        2\sum_{t=13}^T \E\lsb{ \lrb{\frac{5-2E}{8} - \alpha_t( V_{E,1}, \dots, V_{E,2(t-1)} ) }^2 }
    \\
    &\quad
    \overset{\varheartsuit}\ge
        2\sum_{t=13}^T \E\lsb{ \lrb{\frac{5-2E}{8} - \E\lsb{\frac{5-2E}{8} \mid  B_{E,1}, \dots, B_{E,2(t-1)}, D_1,\dots,D_{2(t-1)},U_1,\dots,U_{2(t-1)}}}^2 }
    \\
    &\quad\overset{\vardiamondsuit}=
        2\sum_{t=13}^T \E\lsb{ \lrb{\frac{5-2E}{8} - \E\lsb{\frac{5-2E}{8} \mid  B_{E,1}, \dots, B_{E,2(t-1)}}}^2 }
    =    
        \frac{1}{8}\sum_{t=13}^T \E\lsb{ \lrb{E - \E\lsb{E \mid  B_{E,1}, \dots, B_{E,2(t-1)}}}^2 }
    \\
    &\quad
    \overset{*}=    
        \frac{1}{8}\sum_{t=13}^T \E\lsb{ \lrb{E - \frac{\sum_{s=1}^{2(t-1)}B_{E,s}+1}{2t}}^2 }
    =    
        \frac{1}{8}\sum_{t=13}^T \int_0^1\E\lsb{ \lrb{\e - \frac{\sum_{s=1}^{2(t-1)}B_{\e,s}+1}{2t}}^2 } \dif \e
    \\
    &\quad
    =
        \frac{1}{8}\sum_{t=13}^T \int_0^1\E\lsb{ \lrb{\e - \frac{\sum_{s=1}^{2(t-1)}B_{\e,s}}{2(t-1)}+\frac{\sum_{s=1}^{2(t-1)}B_{\e,s}}{2(t-1)}-\frac{\sum_{s=1}^{2(t-1)}B_{\e,s}}{2t} -\frac{1}{2t}}^2 } \dif \e
    \\
    &\quad
    =
        \frac{1}{8}\sum_{t=13}^T \int_0^1\E\lsb{ \lrb{\e - \frac{\sum_{s=1}^{2(t-1)}B_{\e,s}}{2(t-1)}+\frac{1}{2t(t-1)}\sum_{s=1}^{2(t-1)}B_{\e,s} -\frac{1}{2t}}^2 } \dif \e
    \\
    &\quad
    \ge
        \frac{1}{8}\sum_{t=13}^T \int_0^1\E\lsb{ \lrb{\e - \frac{\sum_{s=1}^{2(t-1)}B_{\e,s}}{2(t-1)}}^2-2\labs{\e - \frac{\sum_{s=1}^{2(t-1)}B_{\e,s}}{2(t-1)}}\labs{\frac{1}{2t(t-1)}\sum_{s=1}^{2(t-1)}B_{\e,s} -\frac{1}{2t}} } \dif \e
    \\
    &\quad
    \ge
        \frac{1}{8}\sum_{t=13}^T \int_0^1\E\lsb{ \lrb{\e - \frac{\sum_{s=1}^{2(t-1)}B_{\e,s}}{2(t-1)}}^2-\frac{1}{t}\labs{\e - \frac{\sum_{s=1}^{2(t-1)}B_{\e,s}}{2(t-1)}} } \dif \e
    \\
    &\quad
    \overset{\star}\ge
        \frac{1}{8}\sum_{t=13}^T \int_0^1 \lrb{\frac{\mathrm{Var}(B_{\e,1})}{2(t-1)}-\frac{1}{t}\sqrt{\frac{\mathrm{Var}(B_{\e,1})}{2(t-1)}} }  \dif \e
    =
        \frac{1}{8}\sum_{t=13}^T \int_0^1 \lrb{\frac{\e(1-\e)}{2(t-1)}-\frac{1}{t}\sqrt{\frac{\e(1-\e)}{2(t-1)}} }  \dif \e
    \\
    &\quad
    =
        \frac{1}{8}\sum_{t=13}^T \lrb{\frac{1}{12(t-1)}-\frac{\pi}{8t\sqrt{2(t-1)}} }
    \ge
        \frac{1}{8}\lrb{\frac{1}{12}-\frac{\pi}{16\sqrt{6}} }\sum_{t=12}^{T-1}\frac{1}{t}
    \ge
        \frac{1}{8}\lrb{\frac{1}{12}-\frac{\pi}{16\sqrt{6}} } \int_{12}^T \frac{1}{s}\dif s
    \\
    &\quad=
        \frac{1}{8}\lrb{\frac{1}{12}-\frac{\pi}{16\sqrt{6}} } \ln\lrb{\frac{T}{12}}
    \ge
        \frac{1}{16}\lrb{\frac{1}{12}-\frac{\pi}{16\sqrt{6}} } \ln(T) \;.
    \end{align*}
    }%
where ``$\spadesuit$'' follows from \cref{l:mysticus}; ``$\circ$'' follows from the fact that $E$ and $V_{\e,1},\dots,V_{\e,2(t-1)}$ are independent of each other; ``$\clubsuit$'' follows from the fact that $\alpha_t$ takes values in $\lsb{\frac{1}{8},\frac{7}{8}}$ and the explicit formula of $F_\e$ in that interval for any $\e \in [0,1]$;  ``$\varheartsuit$'' follows from the fact that $\alpha_t( V_{E,1}, \dots, V_{E,2(t-1)} )$ is $\cF_t \coloneqq \sigma(B_{E,1}, \dots, B_{E,2(t-1)}, D_1,\dots,D_{2(t-1)},U_1,\dots,U_{2(t-1)})$-measurable and that, for any $Y$ the minimizer in $L^2(\cF_t)$ of the functional $X \mapsto \E\lsb{(Y-X)^2}$ is $X = \E[Y \mid \cF_t]$; ``$\vardiamondsuit$'' follows from the fact that $E$ and $(D_1,\dots,D_{2(t-1)},U_1,\dots,U_{2t-1})$ are independent of each other; ``$*$'' follows from the fact $E \mid B_{E,1},\dots,B_{E,2(t-1)}$ has a beta distribution; and ``$\star$'' follows from the fact that $B_{\e,1},B_{\e,2},\dots$ is an i.i.d.\ Bernoulli process of parameter $\e$, together with Jensen's inequality.

\section{Proof of Theorem \ref{t:mbs}}
\label{s:appe:t:mbs}

Without loss of generality, we can (and do!) assume that $T \ge 2$, and so $\log_2(MT) \ge 1$.
First, let $\tau_T$ be the final value of $\tau$ if the algorithm ends at time $T$ without a break, or define it as $\tau-1$ if it ends with a break.
For each $\tau \in [\tau_T]$, we define the epoch $\tau$ as the collection of rounds from $t_{\tau-1}+1$ to $t_{\tau}$.
Notice that, for each $\tau \in [\tau_T]$, we have that $s_{\tau}$ is the number of bits collected during the epoch $\tau$.
Let $Q_1^\star \coloneqq 1/2$, and define by induction $Q_{\tau+1}^\star$ as $Q_{\tau}^\star + \frac{1}{2^{\tau+1}}$ if $F(Q_{\tau}^\star) < 1/2$, as $Q_{\tau}^\star - \frac{1}{2^{\tau+1}}$ if $F(Q_{\tau}^\star) > 1/2$, or as $Q_{\tau}^\star$ if $F(Q_{\tau}^\star) = 1/2$.
If there is $\tau \in \N$ such that $F(Q_\tau^\star) = 1/2$, let $m \coloneqq Q_\tau^\star$.
Otherwise, let $m \in [0,1]$ be such that $F(m) = 1/2$ (its existence has already been pointed out after \Cref{l:mysticus}).
Crucially, notice that for each $\tau \in \N$, we have that $|m-Q_\tau^\star| \le 2^{-\tau}$.

Let $(V_{x,k})_{x\in [0,1], k\in \N}$ be an independent family of random variables with common distribution given by $\nu$, and for each $x \in [0,1]$ and $t \in \N$, define $N_{t}(x) \coloneqq 2 \cdot \sum_{k=1}^{t-1} \I\{P_k = x\}$.
Notice that without loss of generality, we can assume that for each $t \in \N$ it holds that $V_{2t-1} \coloneqq V_{P_t,N_{t}(P_t)+1}$ and $V_{2t} \coloneqq V_{P_t,N_{t}(P_t)+2}$.
Define the ``good'' event
\[
    \cE
\coloneqq
    \bigcap_{i=1}^T \bigcap_{\substack{j=1 \\j \text{ even}}}^{T} \lcb{ \labs{\frac{1}{j} \sum_{k=1}^{j} \I\{V_{Q^\star_i,k} \le Q^\star_i\}  - F(Q^\star_i)} < \sqrt{\frac{\ln(2/\delta)}{2j}} }\;,
\]
and notice that by De Morgan's laws, a union bound, and Hoeffding's inequality, we can upper bound the probability of the ``bad'' event $\cE^c$ by $\Pb[\cE^c] \le \delta T^2 $.
Notice that for each $i,j \in [T]$ with $F(Q_i^\star) \neq \frac{1}{2}$ and $j$ even satisfying $j \ge \frac{2 \ln(2/\delta)}{\lrb{\frac{1}{2}-F(Q_i^\star)}^2}$, then, whenever we are in the good event $\cE$, we have that
\[
    \frac{1}{j}\sum_{k=1}^{j} \I\{V_{Q_i^\star,k} \le Q_i^\star\} + \sqrt{\frac{\ln(2/\delta)}{2j}}
<
    F(Q_i^\star) + \sqrt{\frac{2\ln(2/\delta)}{j}}
\le
    \frac{1}{2}\;,
\]
whenever $F(Q_i^\star) < 1/2$, while
\[
  \frac{1}{j}\sum_{k=1}^{j} \I\{V_{Q_i^\star,k} \le Q_i^\star\} - \sqrt{\frac{\ln(2/\delta)}{2j}}
>
    F(Q_i^\star) - \sqrt{\frac{2\ln(2/\delta)}{j}}
\ge
    \frac{1}{2}\;.
\]
whenever $F(Q_i^\star) > 1/2$.
Instead, if $i,j \in [T]$ with $F(Q_i^\star) = \frac{1}{2}$ and $j$ is even, we have that
\[
    \frac{1}{j}\sum_{k=1}^{j} \I\{V_{Q_i^\star,k} \le Q_i^\star\} + \sqrt{\frac{\ln(2/\delta)}{2j}}
\ge
    F(Q_i^\star)
=
    \frac{1}{2}
\]
and analogously
\[
    \frac{1}{j}\sum_{k=1}^{j} \I\{V_{Q_i^\star,k} \le Q_i^\star\} - \sqrt{\frac{\ln(2/\delta)}{2j}}
\le
    F(Q_i^\star)
=
    \frac{1}{2}\;.
\]
In particular, if we are in the good event $\cE$, these inequalities imply on the one hand that $Q_1 = Q_1^\star, \dots,  Q_{\tau_T} = Q_{\tau_T}^\star$ and, if $\tau \in [\tau_T]$ is such that $F(Q_{\tau}^\star) = 1/2$, then $\tau = \tau_T$.
On the other hand, if $\tau \in [\tau_T]$ is such that $F(Q_\tau^\star) \neq 1/2$ and we are in the good event $\cE$, they imply that the number of bits $s_\tau$ collected during the epoch $\tau$ cannot be greater than $\frac{2 \ln(2/\delta)}{\lrb{\frac{1}{2}-F(Q_\tau^\star)}^2}$, because the condition that ends the epoch $\tau$ with a break is met by the time that we have collected $\frac{2 \ln(2/\delta)}{\lrb{\frac{1}{2}-F(Q_\tau^\star)}^2}$ bits in that epoch.

Define $\tau_T^\# \coloneqq \lceil \log_2(MT)\rceil$, define $\tau_T^\flat$ as the smallest $\tau \in \N$ such that $F(Q^\star_\tau) = 1/2$ if it exists, and $+\infty$ otherwise, and define $\tau^\star_T \coloneqq \min(\tau^\#_T, \tau^\flat_T,\tau_T)$.
In what follows, when we are in the event $\tau^\#_T  > \max(\tau^\flat_T,\tau_T)$, we use the convention that any summation of the form $\sum_{\tau = \tau^\star_T+1}^{\tau_T}$ is zero by definition.
For each $t \in [T]$, define $\cH_t \coloneqq \sigma(V_1,\dots,V_{2t-2})$ as the $\sigma$-algebra generated by the history observed before time $t$.
We can control the regret in the following way
\begin{align*}
    R_T
&=
    \sum_{t=1}^T \E \bsb{ G_t(m) - G_t(P_t) }
=
    \sum_{t=1}^T \E\Bsb{\E \bsb{ G_t(m) - G_t(P_t) \mid \cH_t }}
\\
&\overset{\spadesuit}=
    \sum_{t=1}^T \E\bbsb{\Bsb{\E \bsb{ G_t(m) - G_t(p) }}_{p=P_t} }
\overset{\clubsuit}=
    2\cdot \sum_{t=1}^T \E\lsb{\lrb{\frac{1}{2}-F(P_t)}^2}
\\
&\le
    2 \cdot \sum_{t=1}^T \E \lsb{ \lrb{\frac{1}{2}-F(P_t)}^2 \I_\cE } + \frac{T}{2}\cdot \Pb[\cE^c]
\\
&=
    \E \lsb{ \sum_{\tau=1}^{\tau^\star_T-1}  s_\tau \cdot \lrb{\frac{1}{2}-F(Q_\tau^\star)}^2 \I_\cE }
    + 
    \E \lsb{ \sum_{\tau=\tau^\star_T}^{\tau_T} s_\tau \cdot \lrb{\frac{1}{2}-F(Q_\tau^\star)}^2 \I_\cE }
    +
    \frac{T}{2}\cdot \Pb[\cE^c]
\\
&\overset{\varheartsuit}\le
    \E \lsb{ \sum_{\tau=1}^{\tau^\star_T - 1}  \frac{2 \ln(2/\delta)}{\lrb{\frac{1}{2}-F(Q_\tau^\star)}^2} \cdot \lrb{\frac{1}{2}-F(Q_\tau^\star)}^2 \I_\cE }
    + 
    \E \lsb{ \sum_{\tau=\tau^\star_T}^{\tau_T} s_\tau \cdot M^2 \cdot |m-Q^\star_\tau|^2 }
    +
    \frac{T}{2}\cdot \Pb[\cE^c]
\\
&\le
     (\tau^\#_T-1) \cdot 2\cdot\ln(2/\delta)
    + 
    T \cdot M^2 \cdot 2^{-2\tau^\#_T}
    +
    \delta \cdot \frac{T^3}{2}
\le
    2 + 6 \log_2(MT)\ln(T)\;,
\end{align*}
where in $\spadesuit$ we used the Freezing Lemma (see, e.g.,\cite[Lemma~8]{cesari2021nearest}), in $\clubsuit$ we used \Cref{l:mysticus}, and in $\varheartsuit$ we used that fact that $F(m) = 1/2$ and $F$ is $M$-Lipschitz.

\section{Proof of Theorem \ref{t:lower-bound-2-bit-M}}
\label{s:appe:t:lower-bound-2-bit-M}

We already know that algorithms that have access to full-feedback have to suffer worst-case regret of at least $c_1 \ln T$ if $T \ge c_2$, where $c_1$ and $c_2$ are the constants in the statement of \Cref{t:lower-bound-full-M}. 
In particular, the same statement holds \emph{a fortiori} for any $2$-bit feedback algorithm, given that any $2$-bit feedback algorithm can be trivially converted into an algorithm operating with full-feedback.
It follows that it is enough to prove that there exist two universal constants $\tilde{c}_1$ and $\tilde{c}_2$ such that the worst-case regret of any $2$-bit feedback algorithm is at least $\tilde{c}_1 \ln M$ whenever $T \ge \tilde{c}_2 \log_2(M)$.
In fact, in this case, we can set $\bar c_1 \coloneqq \frac{1}{2}\min(c_1,\tilde{c}_1)$ and $\bar c_2 \coloneqq \max(c_2,\tilde{c}_2)$ to obtain that the worst-case regret of any $2$-bit feedback algorithm is at least $2 \bar c_1 \max(\ln T,\ln M) \ge \bar c_1 \ln (MT)$ whenever $T \ge \bar c_2 \log_2 M$.

We now prove the existence of $\tilde{c}_1$ and $\tilde{c}_2$.
Let $n \in \N$ be the greatest integer such that $2^n \le M$ and consider the elements $\nu_k \in \cD_M$ whose density is $2^n \cdot \I_{(\frac{k-1}{2^n},\frac{k}{2^n})}$ for some $k \in [2^n]$, and notice that the corresponding cdfs are $M$-Lipschitz.

Consider the following surrogate game.
The adversary secretly chooses $k^\star \in [2^n]$. The player action space is $[2^n]$.
The surrogate game ends the first time $t \in \N$ when the player plays $I_t = k^\star$.
Before that, if the player plays $I_t \neq k^\star$, the player suffers a loss $1/2$ and receives $\I\{ I_t \le k^\star\}$ as feedback.
Now, note that we can convert any algorithm $\alpha$ for the $2$-bit feedback problem into an algorithm $\tilde{\alpha}$ for the surrogate game in the following way.
For each $k \in [2^n-1]$, define $J_k \coloneqq [(k-1)2^{-n},k2^{-n})$ and $J_{2^n} \coloneqq [(2^n-1)2^{-n},1]$. Whenever the algorithm $\alpha$ plays $P_t \in J_k$, the algorithm $\tilde{\alpha}$ plays $I_t \coloneqq k$ and passes $\brb{\I\{I_t \le k^\star\},\I\{I_t \le k^\star\}}$ to $\alpha$, where $k^\star$ is the underlying instance of the surrogate game.
Now, notice that we can map every instance $k^\star \in [2^n]$ for the surrogate game into the instance $\nu_{k^\star} \in \cD_M$ of the original problem and that the regret of the algorithm $\alpha$ on the instance $\nu_{k^\star}$ is greater than or equal to than the regret of the algorithm $\tilde{\alpha}$ on the instance $k^\star$.
It follows that a worst-case regret lower bound for the surrogate game is also a worst-case regret lower bound for the original problem.

Fix an algorithm $\alpha$ for the surrogate game.
Given that the surrogate game is deterministic, without any loss of generality we can assume that $\alpha$ is deterministic.
We say that $S \subset [2^n]$ is a discrete segment if $S$ is of the form $\{k \in [2^n] \mid a \le k \le b\}$ for some $a,b \in [2^n]$ with $a<b$.
We can prove the following property by induction on $t = 0,1,\dots,n-1$:
there is a discrete segment $J_t$ with at least $2^{n-t}-1$ elements such that, for each $k,k' \in S_t$, the algorithm has not won the game by the time $t$ and receives the same feedback (and hence selects the same actions) if the underlying instance is $k$ or $k'$.
For $t = 0$ the property is true by setting $S_0 \coloneqq [2^n]$.
Assume that the property is true for some $t \in \{0,1,\dots, n-2\}$.
Assume that $a,b \in [2^n]$ with $a \le b$ are such that $S_t = \{k \in [2^n] \mid a \le k \le b\}$, where $S_t$ is a segment that enjoys the property.
Now, if the algorithm plays $I_{t+1} \notin S_t$ we can set $S_{t+1} \coloneqq S_t$, and we see that the required properties hold trivially.
Instead, if $I_{t+1} \in S_t$ we set $S_{t+1} \coloneqq \{k \in [2^n] \mid I_{t}+1 \le k \le b\}$ if $I_{t+1} < \frac{a+b}{2}$ and we set $S_{t+1} \coloneqq \{ k \in [2^n] \mid a \le k \le I_{t}-1\}$ if $I_{t+1} \ge \frac{a+b}{2}$.
Notice that given that $S_t$ has at least $2^{n-t}-1$ points, we have that $S_{t+1}$ contains at least $\frac{2^{n-t}-2}{2} = 2^{n-(t+1)}-1$ points and, for each $k \in S_{t+1}$, the game does not end by the time $t+1$.
Hence the induction step is proved.
It follows that $S_{n-1}$ is non-empty and, if we pick $k^\star \in S_{n-1}$, the game goes on at least up to time $n-1$ whenever the time horizon $T$ is at least $n-1$.
Hence, if $T \ge \log_2(M)$ (which implies in particular that $T \ge n-1$), the worst-case regret of the algorithm $\alpha$ is at least $\frac{n-1}{2} = \frac{n+1}{2} - 1 \ge \frac{\log_2(M)}{2} - 1 \ge \frac{\log_2(M)}{4} = \frac{1}{4 \ln (2)} \ln M$, where in the last inequality we used $M \ge 16$.
Hence, we can pick $\tilde{c}_1 \coloneqq \frac{1}{4 \ln 2}$ and $\tilde{c}_2 \coloneqq 1$, concluding the proof.

\section{Proof of Theorem \ref{t:ftpsi}}
\label{s:appe-t:ftpsi}

    For any $t \in \N$, tedious but straightforward computations show that
\begin{align*}
    \Pb\lsb{\sup_{p \in [0,1]} \labs{ \Psi(p) - \hat{\Psi}_t(p)} \ge \e}
\le
     \Pb \lsb{ \sup_{p \in \bbR} \labs{ \frac{1}{2t} \sum_{s=1}^{2t} \I\{V_s \le p \} - F(p) } \ge \frac{\e}{4}}
\le
    2 \exp\lrb{- \frac{1}{4} \e^2 t}\;,
\end{align*}
where the last inequality follows from the DKW inequality \cite{massart1990tight}.
Let $p^\star \in \argmax_{p\in [0,1]} \Psi(p)$ (which does exist due to the upper-semicontinuity of $\Psi$).
Then, for any $t \in \N$, we have that
\begin{align*}
    \E\lsb{ \Psi(p^\star) - \Psi(P_{t+1}'') }
&=  
    \E\lsb{ \Psi(p^\star) - \hat{\Psi}_t(p^\star) }
    +
    \E\bsb{ \underbrace{\hat{\Psi}_t(p^\star) - \hat{\Psi}_t(P_{t+1}'')}_{\le 0}  }
    +
    \E\lsb{ \hat{\Psi}_t(P_{t+1}'') - \Psi(P_{t+1}'') }
\\
&\le
    2\E\lsb{ \sup_{p \in [0,1]} \labs{\Psi(p) - \hat{\Psi}_t(p)}}
=
    2\int_0^{+\infty} \Pb \lsb{ \sup_{p \in [0,1]} \labs{\Psi(p) - \hat{\Psi}_t(p)} \ge \e} \dif \e
\\
&\le
    2\int_0^{+\infty} 2 \exp\lrb{ - \frac{1}{4} \e^2 t } \dif \e
=
    \frac{4\sqrt{\pi}}{\sqrt{t}}\;.
\end{align*}
Hence
\[
    R_T
\le
    1
    +
    \E\lsb{ \sum_{t=2}^T \brb{\Psi(p^\star) -\Psi(P_t'')} }
\le
    1
    +
    4 \sqrt{\pi}\sum_{t=1}^{T-1} \frac{1}{\sqrt{t}}
\le
    1
    +
    8 \sqrt{\pi} \cdot \sqrt{T-1}\;. \qedhere
\]

\end{document}